\documentclass[12pt,a4paper]{article}

\usepackage{fullpage}

\usepackage{booktabs} 
\usepackage[ruled]{algorithm2e} 

\SetAlFnt{\small}
\SetAlCapFnt{\small}
\SetAlCapNameFnt{\small}
\SetAlCapHSkip{0pt}
\IncMargin{-\parindent}

\usepackage{colortbl}
\usepackage[numbers]{natbib}
\usepackage{graphicx} 
\usepackage{amsmath}
\usepackage{amssymb}

\usepackage{graphicx} 

\usepackage{times}
\usepackage{soul}
\usepackage{url}
\usepackage[utf8]{inputenc}
\usepackage{graphicx}
\usepackage{amsmath}
\usepackage{amsthm}
\usepackage{booktabs}
 \usepackage{algorithmic}
\usepackage[switch]{lineno}
\usepackage{bigstrut}

\usepackage{amsmath,mathrsfs,amsfonts}
\usepackage{xcolor}
\usepackage{enumitem} 
\usepackage{multirow}

\newtheorem{theorem}{Theorem}

\newtheorem{claim}{Claim}

\newtheorem{corollary}{Corollary}

\newtheorem{definition}{Definition}
\newtheorem{observation}{Observation}
\newtheorem{example}{Example}

\newtheorem{lemma}{Lemma}

\newcommand{\WMMS}{\mathsf{WMMS}}

\newcommand{\cM}{\mathcal{M}}
\newcommand{\cN}{\mathcal{N}}

\newcommand{\cG}{\mathcal{G}}
\newcommand{\cI}{\mathcal{I}}

\newcommand{\R}{\mathbb{R}}

\newcommand{\bA}{\mathcal{A}}

\newcommand{\fw}{\mathbf{w}}
\newcommand{\fv}{\mathbf{v}}

\newcommand{\A}{\mathbf{A}}
\newcommand{\X}{\mathbf{X}}
\newcommand{\B}{\mathbf{B}}

\DeclareMathOperator*{\argmax}{arg\,max}

\newif\ifcomment
\newif\ifchange
\newif\ifresp
\newcommand{\acomment}[1]{{\ifcomment \color{cyan} #1 \fi}}
\newcommand{\hresp}[1]{{\ifresp \color{brown} #1 \fi}}
\newcommand{\achange}[1]{{\ifchange \color{purple} \fi #1}}

\begin{document}
\begin{sloppypar}

\title{Constant Weighted Maximin Share Approximations for Chores\thanks{The authors are ordered alphabetically. Fangxiao Wang is the corresponding author.}}

\author{Bo Li
\hspace{30pt}
Fangxiao Wang
\hspace{30pt}
Shiji Xing 
\medskip
\\
\small
Department of Computing, The Hong Kong Polytechnic University
\\\small
\texttt{comp-bo.li@ polyu.edu.hk}
\\\small
\texttt{\{fangxiao.wang,shi-ji.xing\}@connect.polyu.hk}
}

\date{}
\maketitle

\begin{abstract}
We study the fair allocation of indivisible chores among agents with asymmetric weights. Among the various fairness notions, weighted maximin share (WMMS) stands out as particularly compelling. 
Despite its appeal, the existence of a constant-factor approximation for WMMS has remained an important open problem in weighted fair division [Aziz et al., 2022, Suksompong, 2025]. Prior to our work, the {best-known} approximation ratio was $O(\log n)$, where $n$ is the number of agents. 
In this paper, we make significant progress by presenting the first constant-factor approximation algorithm for WMMS. 
Our main contributions are as follows:

\begin{itemize}
    \item We design the first algorithm that guarantees a 12-approximate WMMS allocation, substantially improving upon the previous $O(\log n)$ upper bound. Our approach introduces a novel analytical framework based on canonical instance reductions, agent delegation, and proxy cost functions to effectively bound agents' costs. Additionally, we provide a polynomial-time implementation for any approximate WMMS algorithm, incurring a factor of 2 loss in the approximation ratio.
    \item We present an improved worst-case lower bound, showing that no algorithm can achieve {a better than 2-approximate WMMS guarantee}, thereby strengthening the previous best lower bound of 1.366.  We further construct a general hard instance, which provides lower bounds for {an arbitrary number of agents}. 
    
    \item Beyond worst-case bounds, we precisely characterize the optimal approximation ratio curve for every possible weight distribution in the two-agent case. Notably, our results imply that a WMMS allocation may not exist for any two agents with different weights, in sharp contrast to the symmetric case where an MMS allocation always exists.
\end{itemize}

\end{abstract}

\section{Introduction}
\label{sec:intro}
Fair allocation of indivisible items has recently attracted considerable interest across fields such as theoretical computer science, computational economics, and multi-agent systems.  
This is partly driven by {practical applications in which} items cannot be divided and must be allocated whole, such as family heirlooms (like jewelry, paintings, and furniture) or office equipment (such as computers, chairs, and printers). 
At the same time, traditional fairness concepts like envy-freeness (EF) \cite{Foley67,varian1973equity} and proportionality (PROP) \cite{Steinhaus48} are not always satisfiable with indivisible items, prompting extensive research into how these criteria can be approximately achieved. 
This has led to the study of relaxed fairness notions, most notably envy-freeness up to one item (EF1) \cite{DBLP:conf/bqgt/Budish10,LiptonMaMo04} and the maximin share (MMS) \cite{DBLP:conf/bqgt/Budish10}, which are now widely accepted in the literature.

Informally, EF1 means that no agent would {prefer another agent's allocation to her own} once a single item is removed either from her own (for chores) or the other agent's bundle (for goods). 
It has been proved that an EF1 allocation always exists for both goods and chores, even when agents are asymmetric -- that is, when they have different weights and the allocation is expected to respect these weights \cite{DBLP:journals/teco/ChakrabortyISZ21,DBLP:conf/sigecom/0001Z023}.
MMS fairness assigns each agent a threshold value, requiring that every agent receive a bundle at least as valuable as this threshold. 
For symmetric agents (where all agents have equal weights), it is known that a constant-factor approximation of MMS fairness can be guaranteed for both goods and chores; see, for example, \cite{DBLP:journals/jacm/KurokawaPW18,DBLP:conf/sigecom/GhodsiHSSY18,DBLP:journals/corr/abs-2510-10423,DBLP:journals/corr/abs-2511-13056} and \cite{DBLP:conf/sigecom/HuangL21,DBLP:conf/sigecom/HuangS23}.
In the case of asymmetric agents, the best possible approximation ratio for MMS -- referred to as WMMS -- in the allocation of goods is $n$, where $n$ is the number of agents \cite{DBLP:journals/jair/FarhadiGHLPSSY19}. 
However, for chores in the asymmetric setting, the optimal approximation ratio for WMMS has remained an open question for a long time. 
This is highlighted as one of the important open problems in the fair division of indivisible items, as noted in surveys \cite{DBLP:journals/sigecom/AzizLMW22,DBLP:journals/ipl/Suksompong25}. 



This is not only an interesting theoretical problem; it also has real-world importance.
For example, volunteers often have varying amounts of time they can offer, so tasks are assigned according to each person's availability. 
Within teams, members may have different workloads depending on their positions, and responsibilities should be distributed in line with each individual's capacity. 
Similarly, universities might be awarded research funding based on student enrollment, research achievements, or societal contributions. 
In all these practical scenarios, the agents are not treated equally. 
Moreover, it is noted in \citep{DBLP:conf/ijcai/00020L24} that, in the context of chore allocation, WMMS can provide a fairer allocation -- in the sense of respecting agents' weights -- compared to other weighted fairness notions such as EF1, APS, PROPX, and CS \cite{DBLP:conf/sigecom/0001Z023,DBLP:conf/sigecom/FeigeH23,DBLP:conf/www/0037L022}, in certain scenarios.
The reason is that all these fairness criteria can be satisfied if a single item is allocated to any agent, regardless of the item's value or the agent's weight. 

For instance, consider an example with two items and two agents, where one agent has a much larger weight than the other. 
The weights and valuations are shown in Table \ref{tab:example:intro}.
For both left and right valuation profiles, allocating item $e_1$ to $a_2$ and $e_2$ to $a_1$ is considered fair according to all the criteria of EF1, APS, PROPX, and CS.  
However, in the left situation, it is more reasonable to allocate $e_1$ to $a_1$ and $e_2$ to $a_2$, as this allocation is exactly proportional to their weights. 
In the right situation, as $\epsilon$ approaches 0, agent $a_1$ is effectively responsible for almost everything, so it would be more reasonable for $a_1$  to receive both items. 
However, any EF1 allocation has to allocate {one item to each agent}\footnote{This is also true for {the stronger notion of envy-freeness up to any item (EFX)} \citep{DBLP:journals/teco/CaragiannisKMPS19}, where the envy from one agent towards another agent can be eliminated after the removal of {\em any} item.}. 
This is because if both items are given to $a_1$, even after removing one item from her bundle, she would still have one item left, while $a_2$ would have none.

\begin{table}[h]
    \centering
    \begin{tabular}{c|c|cc}
     \hline
       agent & weight  & item $e_1$ & item $e_2$ \\
       \hline
        $a_1$ & $w_1 = 1-\epsilon$ &  $1-\epsilon$ & $\epsilon$ \\
        $a_2$ & $w_2 = \epsilon$ &  $1-\epsilon$ & $\epsilon$ \\
        \hline
    \end{tabular}
    \quad \quad
    \begin{tabular}{c|c|cc}
     \hline
       agent & weight  & item $e_1$ & item $e_2$ \\
       \hline
        $a_1$ & $w_1 = 1-\epsilon$ &  $0.5$ & $0.5$ \\
        $a_2$ & $w_2 = \epsilon$ &  $0.5$ & $0.5$ \\
        \hline
    \end{tabular}
    \caption{Allocating two items to two agents, where $\epsilon>0$ is sufficiently small.}  \label{tab:example:intro}
\end{table}

Fortunately, WMMS addresses all of the previously mentioned concerns: in the left situation of Table \ref{tab:example:intro}, the unique WMMS allocation assigns $e_1$ to $a_1$ and $e_2$ to $a_2$, while in the right situation, the only WMMS allocation gives both items to $a_1$ (which will be explained in Section \ref{sec:prelim:concepts}).
In this way, WMMS better reflects the proportionality of agents' weights.
Given the above discussion, designing a constant-approximate WMMS algorithm would be intriguing.

However, prior to our work, it was not known whether constant-factor approximations for WMMS were possible, which limited its practical appeal.  
\citet{DBLP:conf/ijcai/0001C019} were the first to propose an $n$-approximation algorithm and showed that no algorithm can achieve better than a $\frac{4}{3}$-approximation.  
For the case of two agents, \citet{DBLP:conf/ijcai/00020L24} provided an example for which no approximation better than $\frac{\sqrt{3}+1}{2} \approx 1.366$ is possible, and they designed an algorithm that matches this bound.  
As a result, the problem is optimally solved for two agents. 
For more than two agents, the same work introduced two algorithms with $O(\sqrt{n})$ and $O(\log n)$ approximation guarantees, but the existence of a constant-factor approximate WMMS allocation remained unresolved.  

{This paper investigates the problem from two complementary directions: we prove new upper and lower bounds for the WMMS approximation ratio, and we study how asymmetric weights affect the best possible MMS approximation.}


\subsection{Summary of Contribution}

In this section, we provide a brief overview of our main results.

We design the first algorithm that guarantees a 12-approximate WMMS allocation, substantially improving upon the previous $O(\log n)$ upper bound, where $n$ is the number of agents. 
Our method leverages a new analytical framework built on canonical instance reductions, agent delegation, and proxy cost functions to effectively bound agents' costs.
The modularization framework we {develop} not only simplifies the analysis but also offers useful tools for further improving the approximation ratio or designing algorithms in other contexts. 
We regard this result as the major technical contribution of this work.
A technical overview of our approach is provided in Section \ref{sec:techoverview}.
{We note that} the canonical instance reduction may not be polynomial-time computable.
To address this, we present a polynomial-time implementation for any approximate WMMS algorithm, with a {factor-2 loss in the approximation ratio}. 

Secondly, we present a stronger lower bound and prove that no algorithm can achieve {a better than 2-approximate WMMS guarantee}, thereby improving the previous best lower bound of 1.366. 
The design of the hard instance relies on a careful classification of agents and items, as well as geometric distributions of weights and values. 
This lower bound of 2 holds as the number of agents becomes very large.
Together with the impossibility result for two agents proved in \citep{DBLP:conf/ijcai/00020L24}, this implies that WMMS allocations cannot be guaranteed when $n=2$ or when $n$ is very large, leaving it unknown for intermediate values of $n$.  
To bridge this gap, we design a more general instance, which provides a lower bound for {an arbitrary number of agents}.
For example, when $n=2$, it shows that no allocation is better than 1.366-WMMS, which matches the lower bound in \citep{DBLP:conf/ijcai/00020L24}; when $n=3$, no allocation is better than 1.457-WMMS; and as $n$ approaches infinity, the lower bound converges to the golden ratio 1.618.


Finally, we provide a deeper discussion of how varying weight distribution influences the approximation ratio for WMMS.  
For two agents, \citet{DBLP:conf/ijcai/00020L24} showed that the worst-case scenario occurs when $w_1 = \sqrt{3}-1$ and $w_2 = 2-\sqrt{3}$, where certain valuations prevent any allocation from achieving better than a $\frac{\sqrt{3}+1}{2}$ approximation.  
This naturally leads to the question: if the two agents have a different weight distribution, is it possible to design algorithms with improved approximation ratios? 
{At the other extreme}, it is clear that an exact MMS allocation is always possible when the agents have equal weights. 
To address this disparity, we precisely characterize the optimal approximation ratio curve for every possible weight distribution in the two-agent case.  
Our findings also show that a WMMS allocation may not exist for any pair of agents with unequal weights, which stands in stark contrast to the symmetric case. 

\subsection{Other Related Work}
\label{sec:relatedwork}

Fair division has been extensively studied since the seminal work by \citet{Steinhaus48}. 
As proportionality (PROP) \cite{Steinhaus48} and envy-freeness (EF) \cite{Foley67,varian1973equity} may {fail to exist when items are indivisible}, there is ongoing discussion about the most appropriate fairness concepts for such settings. 
Several widely recognized and studied notions have emerged, including maximin share (MMS) \cite{DBLP:conf/bqgt/Budish10}, 
proportional up to one (PROP1) \cite{DBLP:conf/sigecom/ConitzerF017}, proportional up to any (PROPX) \cite{DBLP:journals/orl/AzizMS20}, 
envy-free up to one (EF1) \cite{DBLP:conf/bqgt/Budish10,LiptonMaMo04}, and envy-free up to any (EFX) \cite{DBLP:conf/ecai/GourvesMT14,DBLP:journals/teco/CaragiannisKMPS19}.
Our work is most closely related to MMS fairness, which was originally introduced by \citet{DBLP:conf/bqgt/Budish10} for goods and later extended to chores in \cite{DBLP:conf/aaai/AzizRSW17}. When agents have equal weights, both goods and chores settings admit constant-factor approximations (see, for example, \cite{DBLP:journals/jacm/KurokawaPW18,DBLP:conf/sigecom/GhodsiHSSY18,DBLP:journals/corr/abs-2510-10423,DBLP:journals/corr/abs-2511-13056} for goods, and \cite{DBLP:journals/teco/BarmanK20,DBLP:conf/sigecom/HuangL21,DBLP:conf/sigecom/HuangS23} for chores). 
The best-known approximation for goods is $\frac{7}{9}$ \cite{DBLP:journals/corr/abs-2511-13056}, and for chores, it is $\frac{13}{11}$ \cite{DBLP:conf/sigecom/HuangS23}.
For a more comprehensive overview of fair division with indivisible items, we refer readers to surveys \cite{DBLP:conf/ijcai/AmanatidisBFV22,DBLP:journals/sigecom/AzizLMW22}, and for a detailed discussion on chores, see \cite{guo2023survey}.

%
%

\citet{DBLP:journals/jair/FarhadiGHLPSSY19} were the first to extend the concept of MMS to settings where agents have different weights. 
For goods, they constructed an example showing that no allocation can achieve better than a $\frac{1}{n}$-approximate WMMS, and demonstrated that a simple round-robin algorithm guarantees this $\frac{1}{n}$-approximation. 
Later, \citet{DBLP:conf/ijcai/0001C019} adapted the WMMS concept to chore allocation and proved that the best possible approximation ratio lies between $\frac{4}{3}$ and $n$.
{More recently, \citet{DBLP:conf/ijcai/00020L24} improved both the lower and upper bounds by proving the existence of an $O(\log n)$-approximate WMMS allocation for any number of agents and by showing that the optimal approximation ratio for two agents is 1.366.}
Our work builds on this line of research and further advances these bounds.


Other notions of weighted fairness have also been explored in the literature. \citet{DBLP:conf/sigecom/BabaioffEF21} introduced the AnyPrice Share (APS) as an alternative to MMS and designed a $\frac{3}{5}$-approximation algorithm for goods. 
A 2-approximation algorithm for APS was presented in \cite{DBLP:conf/sigecom/BabaioffEF21} and \cite{DBLP:conf/www/0037L022}, which was later improved to 1.733 by \citet{DBLP:conf/sigecom/FeigeH23}. In the same work, they also proposed the Chore Share (CS) concept to approximate APS for chores.
For goods, a weighted PROP1 allocation is always possible, but a weighted PROPX allocation may not exist \cite{DBLP:journals/orl/AzizMS20}. 
In contrast, for chores, a weighted PROPX allocation (which also implies PROP1) always exists and can be efficiently computed \cite{DBLP:conf/www/0037L022}. For both goods and chores, a weighted EF1 allocation can be computed in polynomial time \cite{DBLP:journals/teco/ChakrabortyISZ21,DBLP:conf/sigecom/0001Z023}, but the existence of weighted EFX allocations is not guaranteed \cite{DBLP:journals/corr/abs-2305-16081}.

Very recently, \citet{DBLP:journals/ipl/Suksompong25} provided a survey of the latest developments and key open questions in the area of weighted fair division of indivisible items.

Finally, we note that {in simultaneous work} \cite{DBLP:journals/corr/abs-2510-10698}, another constant-factor approximate WMMS algorithm was obtained independently using a different approach. 
For comparison, their algorithm achieves an approximation ratio of 20, which is somewhat weaker than ours. 
Moreover, their study does not {involve lower bounds or examine} the approximation ratio across different weight distributions.

\section{Definitions and Technical Overview}

For any positive integer $k$, denote by $[k] = \{1, \ldots, k\}$.
We consider a fundamental fair allocation model by allocating $m$ indivisible chores $\cM=\{e_1,\ldots, e_m\}$ to $n$ agents $\cN=\{a_1,\ldots,a_n\}$.
Each agent $a_{i}$ has a cost function (or valuation) $v_{i}:2^{\cM}\rightarrow \R^+\cup\{0\}$, where $v_{i}(S)$ represents her cost of completing the items in $S \subseteq \cM$. 
For simplicity, when a bundle contains a single item $e$, denote $v_{i}(e) = v_{i}(\{e\})$.
We assume the cost functions are additive, i.e., $v_{i}(S)=\sum_{e\in S}v_{i}(e)$ for every $S\subseteq \cM$.
We sometimes assume, without loss of generality, that the valuations are normalized, i.e., $v_{i}(\cM) =1$ for all agents $a_{i}$.
An {allocation}, denoted by $\A=(A_{1},\ldots,A_{n})$, is an $n$-partition of $\cM$, where $A_{i}$ is the set of items allocated to agent $a_{i}$,
such that $A_1\cup\cdots\cup A_n= \cM$ and $A_{i}\cap A_{j}= \emptyset$ for all $i\neq j$.
Denote by $\bA$ the set of all allocations. 
An allocation is called {partial} if $A_1\cup\cdots\cup A_n \subsetneq\cM$.

\subsection{Solution Concepts}
\label{sec:prelim:concepts}
We are interested in the setting when agents are asymmetric. 
Formally, each agent $a_i$ has a weight $0<w_i <1$, representing her share in completing all items.
Assume $w_1\ge w_2\ge \cdots \ge w_n$. 
For simplicity, we normalize the agents' weights such that $\sum_{i=1}^n w_i = 1$. 
Letting $\fv = (v_{1},\ldots,v_{n})$ and $\fw = (w_{1},\ldots,w_{n})$, a chore allocation instance is represented by $\cI=(\cN,\cM,\fw,\fv)$.
The items are expected to be allocated in a way that respects the agents' weights.
Formally, an allocation $(A_1,\ldots,A_n)$ is proportionally (PROP) fair if $v_i(A_i) \le w_i\cdot v_i(\cM)$ {for every} agent $a_i$. 
However, achieving PROP is often impossible; for example, if a single item with non-zero cost must be allocated between two agents, whichever agent receives the item, the allocation will not be PROP for her.
Given that items may be indivisible, the fairest allocation -- one that is as proportional to the agents' weights as possible -- for agent $a_i$ is the one that minimizes the maximum ratio ({also called the unfairness ratio}) of excess across all agents, i.e., 
\[
\min_{\A\in \bA} \max_{j\in [n]}\frac{v_i(A_j)}{w_j}.
\]
Then, the weighted maximin share (WMMS) of $a_i$ is {her weight multiplied by the minimum unfairness ratio above}, i.e., 
\[
\WMMS_i=w_i \cdot \min_{\A\in \bA}\max_{j \in [n]} \frac{v_i(A_j)}{w_j}.
\]
It is clear that $\WMMS_{i} \ge w_i \cdot v_{i}(\cM)$ and thus {WMMS relaxes} the requirement of PROP. 
It is worth noting that the definition of WMMS does not require the weights to be normalized.
A partition $\A=(A_{1},\ldots,A_{n})$ is called a WMMS-defining partition for agent $a_{i}$, if for all $j \in [n]$,
\[
\frac{v_{i}(A_{j})}{w_j} \leq \frac{\WMMS_{i}}{ w_i}.
\]
\begin{definition}[$\alpha$-WMMS]
   Given $\alpha \ge 1$, an allocation $\A=(A_{1},\ldots,A_{n})$ is called $\alpha$-approximate WMMS fair ($\alpha$-WMMS) if 
    $v_{i}(A_{i})\leq \alpha\cdot \WMMS_{i}$ for all agents $a_{i}\in \cN$. 
    If $\alpha = 1$, the allocation is called {WMMS-fair}. 
\end{definition}



\begin{example}
\label{example:1}
Consider an example with two agents and two items, where the weights and valuations are shown in Table \ref{tab:example:intro} of Section \ref{sec:intro}.
For the left valuation profile, WMMS coincides with PROP.
For the right valuation profile, it can be verified that, given sufficiently small $\epsilon > 0$,
\begin{align*}
    \WMMS_1 &= (1-\epsilon)\cdot \min \left\{\frac{v_1(\{e_1,e_2\})}{1-\epsilon}, \frac{v_1(e_1)}{\epsilon}\right\} = 1,
\end{align*}
and
\begin{align*}
    \WMMS_2 &= \epsilon\cdot \min \left\{\frac{v_2(\{e_1,e_2\})}{1-\epsilon}, \frac{v_2(e_1)}{\epsilon}\right\} =\frac{\epsilon}{1-\epsilon} \to \epsilon.
\end{align*}
By the definition of $\WMMS_i$, we need to enumerate all possible allocations to find the smallest unfairness ratio. 
However, in this example, it is clear that this ratio must appear between allocating both items to agent $a_1$ and allocating one of the items, say $e_1$, to agent $a_2$.
By the design of the weights, a WMMS allocation should allocate both items to $a_1$.
In fact, if $\epsilon \to 0$, any finite-approximate WMMS allocation would always allocate both items to $a_1$, which coincides {with} our intuition that an agent with negligible weight should not be assigned challenging tasks. 
However, allocating each agent one item satisfies other weighted fairness notions, such as APS, CS, EF1, and PROP1. 


\end{example}

It is widely known that, for additive valuations, to approximate WMMS, 
it suffices to focus on the case when all agents agree on the same ordering of the items according to their values. 
Formally, an instance is called {\em identical ordering} (IDO) if $v_i(e_1) \ge v_i(e_2)\ge \cdots\ge v_i(e_m)$ for all agents $a_i \in \cN$.
It is proved, for example, in \cite{DBLP:journals/aamas/BouveretL16,DBLP:journals/teco/BarmanK20}, that if we can compute an $\alpha$-WMMS allocation for all IDO instances, then we can compute an $\alpha$-WMMS allocation for all instances that may not be IDO. 
Therefore, throughout this paper, we only consider IDO instances. 

\subsection{Technical Overview}
\label{sec:techoverview}

Example \ref{example:1} shows that allocating items to an agent with a very small weight might severely harm the approximation ratio. 
This difficulty also explains why obtaining a constant-factor approximation for WMMS is more difficult than for other fairness notions like APS and CS. {For APS, CS, and related concepts such as EF1 and PROP1, the corresponding guarantees can be satisfied by allocating a single item to an agent, regardless of the item's value or the agent's weight.} 
To overcome this difficulty, \citet{DBLP:conf/ijcai/00020L24} proved that, for all $1\le i <n$, each item after the index $\lfloor\frac{w_1}{w_{i+1}}\rfloor + \cdots + \lfloor\frac{w_i}{w_{i+1}}\rfloor$ has value no more than $\WMMS_i$ for agent $a_i$.
Based on this, their algorithm only allocates items after this index to $a_i$.
However, this approach results in agent $a_1$ receiving more items and potentially incurring a cost as high as $\Omega(\log n) \cdot\WMMS_1$. 
To achieve a better approximation, it is necessary to start allocating items to the other agents earlier. {Implementing and analyzing this idea, however, is technically more challenging.} 

To this end, we propose a modularization framework and decompose the problem into several components -- canonical instances, delegating agents, and proxy cost functions -- that make the analysis more manageable.
We first introduce the concept of {\em canonical} instances. Informally, in a canonical instance, each agent's weight and each item's value are powers of $\frac{1}{2}$ times $w_1$ (the largest agent weight), and every agent's WMMS is equal to her weight. 
The structured properties of canonical instances greatly simplify the analysis, allowing us to effectively bound both the value and the number of items in our algorithm. 
We show that reducing any arbitrary instance to a canonical one increases the approximation ratio of WMMS by at most a factor of 4.


We then design an algorithm for canonical instances, which ensures 3-WMMS.
The algorithm processes items one at a time in decreasing order of value. 
Each item can only be given to an agent for whom the item's value does not exceed her WMMS; such agents are referred to as {\em active} agents. 
If multiple active agents are available, the item is allocated to the agent with the smallest value for it. If there are multiple agents with the same smallest value for the item, then it is allocated to the agent with the smallest weight.
An agent $a_i$ exits the algorithm (and becomes {\em inactive}) if her value for her allocated items is no smaller than $3\WMMS_i$.
Due to the properties of canonical instances, this value is exactly $3\WMMS_i$ but cannot exceed it, ensuring the approximation guarantee for exiting agents. 
The aforementioned tie-breaking rule plays a critical role: when an agent with a larger weight exits the algorithm, she must have been allocated a bundle of items whose value is very large for agents with smaller weights. 
To prove the correctness of the algorithm, we need to prove that, for every item, there is always at least one active agent available, which is the most technically involved aspect of the algorithm.
In this regard, we introduce {\em delegating} agents and prove that the {activity windows} of all delegating agents can span all items via their {\em proxy cost} -- an upper bound {on} their actual cost.

Taking into account the loss incurred during the reduction to canonical instances, our algorithm achieves a 12-WMMS guarantee.
Our result is primarily existential in nature. While the aforementioned 3-WMMS algorithm runs in polynomial time, the reduction to canonical instances does not, as it requires computing the WMMS partitions for the agents and may generate an exponential number of items.
To address this, we present an alternative linear programming algorithm that runs in polynomial time. We prove that, provided the instance admits an $\alpha$-WMMS allocation, our algorithm can compute a 
$(2\alpha+\epsilon)$-approximate WMMS allocation, where $\epsilon$ is arbitrarily small. Our analysis leverages the rounding technique introduced in \cite{DBLP:journals/mp/ShmoysT93} and refines the analysis from \cite{DBLP:conf/ijcai/0001C019}.
Combined with our existential result, a $(24+\epsilon)$-WMMS allocation can be found in polynomial time.

\section{Canonical Instance and Polynomial-time Implementation}\label{sec:can}

In this section, we introduce canonical instances.
We show that any instance can be reduced to a canonical instance, with the WMMS approximation ratio increasing by at most a factor of 4. 

%

\begin{definition}[Canonical Instances]
    An instance $\cI=(\cN,\cM,\fw,\fv)$ is called \textit{canonical} if it satisfies the following properties:
    \begin{itemize}
           \item {\em Normalized Weights:} $w_1\ge w_2\ge \cdots\ge w_n$, $w_1+\cdots+w_n = 1$, and for all $i =1,\ldots, n$, $w_i=\frac{w_1}{2^p}$ for some integer $p\ge0$; 
           \item {\em Normalized Valuations:} for all agents, $v_{i}(\cM)=1$, and for every $e\in \cM$, $v_{i}(e)=\frac{w_1}{2^p}$ for some integer $p\ge0$; 
           \item {\em IDO:} $v_i(e_1) \ge v_i(e_2) \ge \cdots \ge v_i(e_m)$ for all agents $a_i$;
           \item {\em Proportional:} $\WMMS_{i}=w_i$ for all agents $a_i$.
    \end{itemize}
    
    
\end{definition}

\begin{theorem}\label{the:can}
If there is an $\alpha$-WMMS allocation for all canonical instances, there is a $4\alpha$-WMMS allocation for arbitrary instances.
\end{theorem}

We provide the proof of Theorem~\ref{the:can} in Section~\ref{prof:sec:can}.
Informally, IDO and proportional properties can be satisfied without loss of generality, whereas enforcing either the normalized weights or normalized valuations properties results in a {factor-2 loss}. 

As mentioned, the transformation from an arbitrary instance to a canonical one may take exponential time, since the reduction requires the WMMS partitions of the agents and may construct an exponential number of auxiliary items. 
We next prove that, as long as there exists an $\alpha$-WMMS allocation for the instance, a $(2\alpha+\epsilon)$-WMMS allocation can be computed in polynomial time, for any {arbitrarily small} $\epsilon>0$. 

\begin{theorem}
\label{thm:polytime}
    Given any instance and a constant $\epsilon>0$, if there exists an $\alpha$-WMMS allocation, then a $(2\alpha+\epsilon)$-WMMS allocation can be computed in polynomial time.
\end{theorem}

The algorithm in Theorem \ref{thm:polytime}
leverages the rounding technique introduced in \cite{DBLP:journals/mp/ShmoysT93} and {follows the analysis of} \cite{DBLP:conf/ijcai/0001C019} with minor refinements.  Further details are provided in Appendix \ref{app:poly:proof}.




\section{A 3-WMMS Algorithm for Canonical Instances}\label{sec:main}

The algorithm is simple, as described in Algorithm \ref{alg:main}.
Let $\cI=(\cN, \cM, \fw,\fv)$ be a canonical instance.
We begin by introducing additional notation for ease of presentation.
Since the agents have normalized weights, we group them in $k$ groups by their weights ordered decreasingly $\cG=\{G_1,\ldots,G_k\}$, where the agents in the same group $G_i$ have the same weight $w_i$ and $w_i>w_j$ for $1\le i<j\le k$.
Group $G_i$ contains $n_i$ agents who are renamed as $G_i=\{a_{i,1},\ldots,a_{i,n_i}\}$.
{For valuations, allocations, and WMMS values, we use two subscripts: $v_{i,j}$, $A_{i,j}$, and $\WMMS_{i,j}$ refer to the $j$-th agent in group $G_i$.}
We also shorten $\fw$ by removing repeated entries, i.e., $\fw=(w_1,\ldots,w_k)$.
Let $W_i=n_i w_i$ and thus $\sum_{j\in [k]} W_j= 1$.
A canonical instance is also denoted by $\cI = (\cG, \cM, \fw, \fv)$.

\subsection{Main Result and the Algorithm}

In Algorithm \ref{alg:main}, the items are allocated one by one from the largest to the smallest. 
{For each item $e_h$, define the set of {\em active} agents as $\cN_h = \{a_{i,j} \in \cN \mid v_{i,j}(e_h)\leq w_i\}$; these are the agents who have not yet exited the algorithm and whose cost for $e_h$ does not exceed their WMMS.}
The item is allocated to the agent $a_{i,j}$ who has the smallest value for $e_h$, and if there is a tie, the algorithm chooses the agent who has the smallest weight. 
If there is still a tie, it can be broken arbitrarily.  
{When the cost of $a_{i,j}$'s bundle reaches $3\WMMS_{i,j}$, $a_{i,j}$ becomes {\em inactive} and exits the algorithm.}
The main technical difficulty is to prove that {each time an item $e_h$ is to be allocated}, $\cN_h$ is not empty, so all items can be allocated.

\begin{algorithm}
\caption{Computing 3-WMMS Allocations for Canonical Instances}\label{alg:main}
\textbf{Input:} A canonical instance $(\cN, \cM, \fw,\fv)$.\\
\textbf{Output:} A 3-WMMS allocation $(A_{1,1},\ldots,A_{k,n_k})$.
\begin{algorithmic}[1] 
\STATE Initialize $A_{i,j} \leftarrow \emptyset$ for all agents $ a_{i,j}\in \cN$.\\
\FOR{each item $e_h$ with $h = 1, \ldots, m$}{
\STATE Find the set of active agents $\cN_h=\{a_{i,j} \in \cN \mid v_{i,j}(e_h)\leq w_i\}$.\\
\STATE Let $a_{i,j}$ be one agent in $\cN_h$ who has the smallest value for $e_h$ and has the smallest weight.\\
\STATE $A_{i,j}\leftarrow A_{i,j}\cup \{e_h\}$.\\
\IF{$v_{i,j}(A_{i,j})\geq 3w_i$}{
\STATE $\cN\leftarrow \cN\setminus \{a_{i,j}\}$.
}
\ENDIF
}
\ENDFOR
\end{algorithmic}
\textbf{return} $(A_{1,1}, \ldots, A_{k,n_k})$
\end{algorithm}

\begin{theorem}\label{the:main}
For any canonical instance $(\cN, \cM, \fw,\fv)$, Algorithm \ref{alg:main} returns a 3-WMMS allocation.
\end{theorem}

Combining Theorem \ref{the:main} with Theorem \ref{the:can}, a 12-WMMS allocation is guaranteed to exist.
Further by Theorem \ref{thm:polytime}, for any small constant $\epsilon>0$, a $(24+\epsilon)$-WMMS allocation can be found in polynomial time.

    

In the remainder of this section, we prove Theorem \ref{the:main}. 
We begin by introducing additional definitions and technical claims that will be used to analyze the algorithm. 
For ease of presentation, we use $[h_1, h_2]$ to denote the set of items from $e_{h_1}$ to $e_{h_2}$ and thus $v_i([h_1, h_2]) = v_i(\{e_{h_1}, \ldots, e_{h_2}\})$.
By the $h$-th round of Algorithm \ref{alg:main}, we mean the moment when the algorithm allocates item $e_h$.
{Let $A_{i,j}^h$ denote the partial bundle assigned to agent $a_{i,j}$ at the end of round $h$, and let $sc_h(S)=\sum_{a_{i,j}\in S}v_{i,j}(A_{i,j}^h)$ denote the total cost of agents in $S \subseteq \cN$.}
Denote $L_1=0$ and $L_r=\frac{\sum_{j<r}W_j }{w_r}$ for $r = 2,\ldots,k$.
Recall that $k$ is the number of groups.
{Claim~\ref{cla:lr}, proved in Appendix~\ref{sec:appendix:3WMMS}, shows that every agent has cost at most $w_r$ for each item after index $L_r$.}

\begin{claim}\label{cla:lr}
    For any agent $a_{i,j}$ and item $e_h\in\cM$ with $h> L_r$ and $r\in [k]$, we have $v_{i,j}(e_h)\leq w_r$.
\end{claim}

{We next prove that the canonical-instance structure ensures that no agent $a_{i,j}$ receives a bundle whose total cost exceeds $3w_i$.}

\begin{claim} \label{claim:exact}
    The allocation $(A_{1,1},\ldots, A_{k,n_k})$ computed by Algorithm~\ref{alg:main} guarantees that $v_{i,j}(A_{i,j})\leq 3w_i$ for every agent $a_{i,j}\in \cN$.
\end{claim}

\begin{proof}
    We use $A_{i,j} =\{o_1,\ldots,o_t\}$ to denote the allocation of agent $a_{i,j}$ at the end of the algorithm. 
    {By construction, $a_{i,j}$ obtains $o_h$ only if $v_{i,j}(o_h) \leq w_i$ and the cost of its previously allocated items is less than $3w_i$.}
    Thus, to prove the lemma, it suffices to show that after allocating $o_h$ to $a_{i,j}$, $3w_i - v_{i,j}([1,h])$ is a non-negative multiple of $v_{i,j}(o_h)$.    

    We prove the above claim by induction. Consider the first item $o_1$ allocated to $a_{i,j}$. 
    By the definition of canonical instance, we have $v_{i,j}(o_1) = \frac{w_i}{2^p}$ for a non-negative integer $p$.
    {Then,}
    \[
        \frac{3w_i-v_{i,j}(o_1)}{v_{i,j}(o_1)} = (3w_i - \frac{w_i}{2^p}) / \frac{w_i}{2^p} = 3\cdot 2^p -1,
    \]
    which means that the difference is an integral multiple of $v_{i,j}(o_1)$.
    Assume that the claim holds until $o_{h-1}$ and $3w_i - v_{i,j}([1,h-1]) = q\cdot v_{i,j}(o_{h-1})$ for some integer $q$.
    Since the instance is canonical, $v_{i,j}(o_h) = \frac{v_{i,j}(o_{h-1})}{2^p}$ for some non-negative integer $p$. Therefore, after allocating $o_h$ to $a_{i,j}$, we have 
    \[
        \frac{3w_i-v_{i,j}([1,h])}{v_{i,j}(o_h)} = (q\cdot v_{i,j}(o_{h-1}) - \frac{v_{i,j}(o_{h-1})}{2^p}) / \frac{v_{i,j}(o_{h-1})}{2^p} = q\cdot 2^p -1,
    \]
    {which proves the induction step.}
    Moreover, the claim indicates that $3w_i - v_{i,j}(A_{i,j}) \geq 0$ throughout the algorithm, which proves the claim.
\end{proof}

By Claim \ref{claim:exact}, Algorithm \ref{alg:main} does not allocate an agent $a_{i,j}$ a bundle with cost higher than $3w_i$.
It remains to prove that all items can be allocated.
As mentioned earlier, we only need to prove that in every round of the algorithm, there is at least one active agent, i.e., $\cN_h \neq \emptyset$.
In the following, we first introduce delegating agents.

For any $1\le i\le k$, denote by $a_i$ the agent in group $G_i$ who is the last to become inactive.
Let $\cN' = \{a_i \mid 1\le i\le k\}$ be the set of these agents.
Thus, we ignore the subscript of $j$ for agents in $\cN'$ to simplify the notations.
{For each $a_i\in \cN'$, let $d_i = \argmax_h\{v_i(e_h)>w_i\}$ be the largest index of an item whose cost is greater than $w_i$ under $v_i(\cdot)$, and let $f_i$ denote the time when $a_i$ becomes inactive.}
Thus the life window of $a_{i}$ is $\Phi_{i}= [d_{i}+1, f_{i}]$.
By Claim~\ref{cla:lr}, 
$d_{i}\leq L_i$.

\begin{claim}
    For any two agents $a_{i_1}$ and $a_{i_2}$ such that $i_1 < i_2$ and $\Phi_{i_1} \cap \Phi_{i_2} \neq \emptyset$, if an item $e_h \in \Phi_{i_1} \cap \Phi_{i_2}$ is allocated to $a_{i_1}$, then $v_{i_2}(e_h) \geq 2 \cdot v_{i_1}(e_h)$.
\end{claim}

Next, we iteratively construct the set of {\em delegating agents} $\cN^* = \{a_{r_1}, \ldots,a_{r_t}\}\subseteq \cN'$. Initializing $\cN^* = \{a_1\}, i=1$,
repeat the following procedure until 
$i = k$:
\begin{enumerate}
    \item Set $i \gets i+1$;
    \item If $f_i > \max_{j\in {\cN^*}}f_j$, set $\cN^* \gets \cN^* \cup \{i\}$;
    \item For all $i,j \in \cN^*, i<j$, if $d_i \ge d_j$, set $\cN^* \gets \cN^* \setminus \{i\}$.
\end{enumerate}
After the set $\cN^*$ is finalized by the above routine, we sort the agents in ascending order and name them as  $\cN^* = \{a_{r_1}, \ldots,a_{r_t}\}$.
The delegating agents are illustrated in Figure \ref{fig:proxies}.
{The construction implies the following ordering property.}
\acomment{Follow-up author verification: The construction maximizes the exit time $f_l$ among eligible life windows, but the written conditions do not rule out selecting an eligible group index $l<r_i$ with a later exit time. The later intervals $G_{r_i}\cup\cdots\cup G_{r_{i+1}-1}$ require $r_i<r_{i+1}$. Please add the short argument proving this group-index monotonicity, or state it as an additional claim.}

\begin{claim}
\label{cla:sel}
    $d_{r_i} < d_{r_j}$ and $f_{r_i}< f_{r_j}$ for all $t\ge j>i>0$.
\end{claim}
\begin{figure}[h]
        \centering
        \includegraphics[width=1\linewidth]{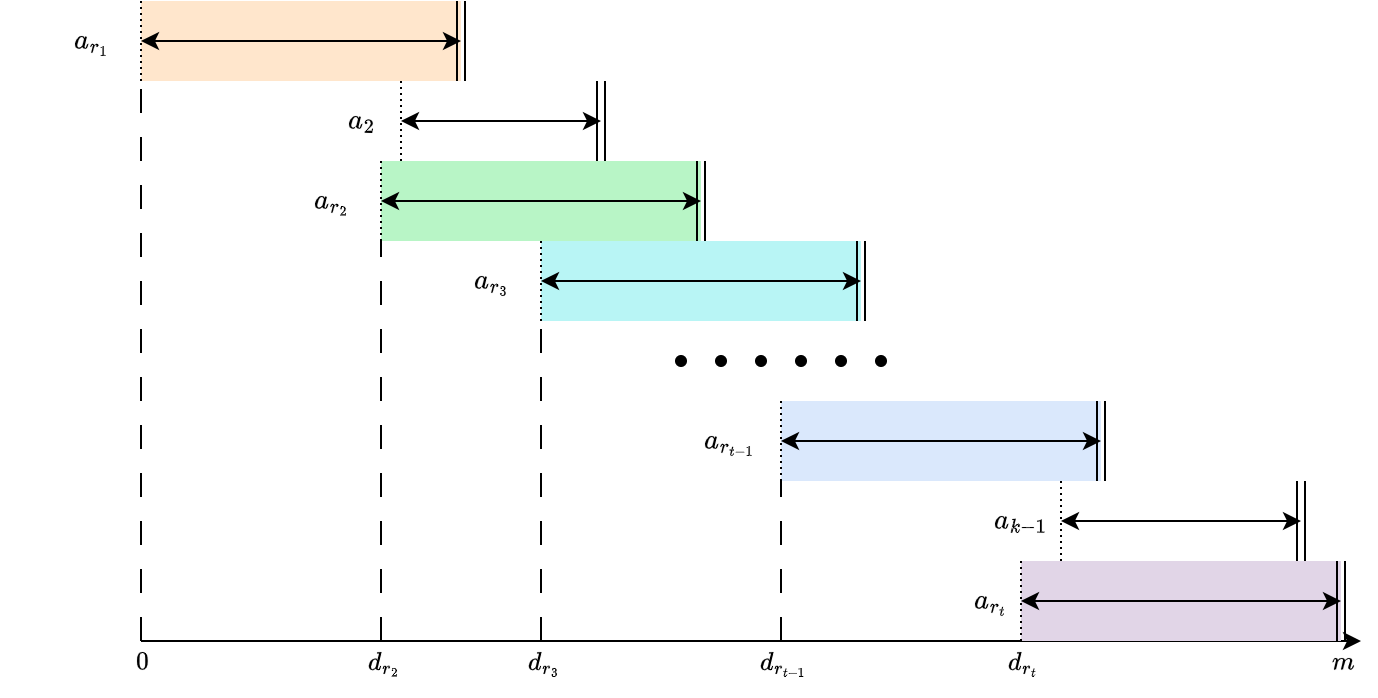}
        \caption{An illustration of delegating agents $\cN^*$.
        For each interval $a_{r_i}$, the left dashed boundary means that $a_i$ becomes active at item $d_{r_i}+1$, and the right double lines mean $a_i$ becomes inactive after receiving this item.}
        \label{fig:proxies}
    \end{figure}

{We now prove that the collective life window of the delegating agents covers all items; together with Claim~\ref{claim:exact}, this proves Theorem~\ref{the:main}.}

\begin{lemma}
\label{lem:allocated} 
$\bigcup_{a_i\in \cN^*} \Phi_{i} = [1,m]$.
\end{lemma}

\subsection{Proof of Lemma \ref{lem:allocated}}

{We prove Lemma~\ref{lem:allocated} by contradiction. Suppose that $e_h$ is the first unallocated item.}
Denote by $p$ the largest group index such that $L_p\leq h-1$.
By Claim~\ref{cla:lr}, every agent $a_{i,j}$ in $G_1 \cup \cdots \cup G_p$ has value no greater than $w_i$ for item $e_h$. 
{If $sc_{h-1}(\cN)<3W_1+\cdots+3W_p$, then at least one agent in $G_1 \cup \cdots \cup G_p$ can still receive $e_h$. Hence $\cN_h \neq \emptyset$, contradicting the choice of $e_h$.}

\begin{claim}
    $sc_{h-1}(\cN) < 3W_1 + \ldots + 3W_p$. 
\end{claim}

    We will use the valuations of the delegating agents to bound $sc_{h-1}(\cN)$.
    First, we construct the set of delegating agents $\cN^* = \{a_{r_1}, \ldots,a_{r_t}\}\subseteq \cN'$ where $a_{r_t}$ is the last agent who is active in $G_1\cup\cdots \cup G_p$.
    {For example, agent $a_{r_1}$ is active during the allocation of items $[1,d_{r_2}]$.}
    Since the algorithm tries to allocate items to whoever values them the least, for each item $e\in[1,d_{r_2}]$, it is allocated to an agent whose cost of $e$ is at most $v_{r_1}(e)$.
    As a consequence, the social cost incurred by allocating $[1,d_{r_2}]$ to agents is upper bounded by $v_{r_1}([1,d_{r_2}])$. 
    {Similarly, the social cost of the items in $[d_{r_i}+1,d_{r_{i+1}}]$ can be bounded by $v_{r_i}([d_{r_i}+1,d_{r_{i+1}}])$.}
    Therefore, we have the following upper bound:
    \[
     sc_{h-1}(\cN) \leq v_{r_1}([1, d_{r_2}]) + v_{r_2}([d_{r_2}+1, d_{r_3}]) + \cdots + v_{r_t}([d_{r_t}+1, h-1]).
    \]
    Unfortunately, this bound is too loose to prove the claim. Further observations are necessary to tighten it.
    
    Note that, after round $d_{r_{i+1}}$, some agents in $G_{r_i} \cup \cdots \cup G_{r_{i+1}-1}$ are still active.
    For example, as illustrated in Figure \ref{fig:proxies}, agent $a_{r_1}$ is still active in rounds $[d_{r_2}+1,f_{r_1}]$, and will be allocated more items by the algorithm until she eventually becomes inactive.
    Consider an item $e_j$ with $d_{r_{i+1}}+1\leq j$ being allocated to one agent in $G_{r_i} \cup \cdots \cup G_{r_{i+1}-1}$.
    Then, in the above upper bound, the social cost increased by allocating $e_j$ is estimated as $v_{r_{i'}}(e_j)$ with $i'> i$ while it should be $v_{r_i}(e_j)$. 
    Due to the construction of the canonical instance, $v_{r_i}(e_j) \leq \frac{1}{2} v_{r_{i'}}(e_j)$.
    Therefore, we can safely subtract $v_{r_{i'}}(e_j) - v_{r_i}(e_j) \geq v_{r_i}(e_h)$ from the upper bound and maintain the inequality. The same can be applied to 
    each item allocated to the agents  in $G_{r_i} \cup \cdots \cup G_{r_{i+1}-1}$ after round $d_{r_{i+1}}$.
    {In other words, for each agent in $G_{r_i} \cup \cdots \cup G_{r_{i+1}-1}$, the increase in the cost of that agent's bundle between round $d_{r_{i+1}}$ and the moment the agent becomes inactive can be subtracted from the upper bound above.}

{
    For $i=1,\ldots,t-1$, let
    \[
    Y_i = 3(W_{r_i} + W_{r_i+1} + \cdots + W_{r_{i+1}-1}) - sc_{d_{r_{i+1}}}(G_{r_i} \cup \cdots \cup G_{r_{i+1}-1})
    \]
    denote the increment of the total cost of agents in $G_{r_i} \cup \cdots \cup G_{r_{i+1}-1}$ from the end of round $d_{r_{i+1}}$ to the round when all of them become inactive (when their total cost becomes $3(W_{r_i} + W_{r_i+1} + \cdots + W_{r_{i+1}-1})$ by Claim \ref{claim:exact}). 
    When those agents exit, all delegating agents $a_{r_{i'}}$ with $i'>i$ are still in the algorithm.}
    We now have the refined upper bound:
    \begin{equation}
     sc_{h-1}(\cN) \leq v_{r_1}([1, d_{r_2}]) + v_{r_2}([d_{r_2}+1, d_{r_3}]) + \cdots + v_{r_t}([d_{r_t}+1, h-1])- \sum_{i\in [t-1]} Y_i. \label{eq:swh-1:1}
    \end{equation}

{We next introduce the {\em proxy cost function} $v^{\star}_i(\cdot)$ for each delegating agent $a_i$; this function upper-bounds the real cost function of $a_i$ on the intervals used below.}
Formally, for $a_i\in \cN^*$, $v^{\star}_i(\cdot)$ is defined as follows:
\begin{align*}
    v^{\star}_i(e_h)=
    \begin{cases}
     w_i, & h\leq L_i+n_i,\\
     w_{i+1}, & L_i+n_i<h\leq L_i+n_i+n_{i+1},\\
     ~\vdots\\
     w_k, &L_i+n_i+\cdots+n_{k-1} < h \le L_i+n_i+\cdots+n_{k},\\
     0, &h>L_i+n_i+\cdots+n_{k}.
    \end{cases}
\end{align*}

Accordingly, we have the following claim. 

\begin{claim}\label{cla:star}
    {For any cost function $v_i(\cdot)$ and $j>d_i$, we have}
    \[
    v_{i}([{d_{i}+1},j]) \le v^*_{i}([{d_{i}+1},j]).
    \]  
\end{claim}

\begin{proof}
    In the following, we only prove for the case when $i \geq 2$, and the proof for the case when $i=1$ is analogous.
    {Assume for contradiction that $s$ is the smallest item index with $d_i<s\leq j$ such that
    $$\sum_{h\in \{d_i+1,\ldots, s\}}v_i(e_h)>\sum_{h\in \{d_i+1,\ldots, s\}}v^\star_i(e_h).$$}
    Also, since the above inequality does not hold for $s-1$, we have 
    $v_{i}(e_s) > v^\star_{i}(e_s)$.
    Because $s> d_i$, $v_{i}(e_s) \leq w_i$. The above inequalities combined indicate that $s> d_i+n_i$.

    \achange{We first rule out the case $v^\star_i(e_s)=0$. In this case, by the definition of $v^\star_i$, all proxy costs after $s$ are also zero, and hence
    \[
        \sum_{q=1}^{s} v^\star_i(e_q)=\sum_{q=1}^{|\cM|} v^\star_i(e_q).
    \]
    Combining the choice of $s$ with the fact that $\sum_{q=1}^{d_i} v_i(e_q)\geq \sum_{q=1}^{d_i}v^\star_i(e_q)$ gives
    \[
        \sum_{q=1}^{s}v_i(e_q)>\sum_{q=1}^{s}v^\star_i(e_q)=\sum_{q=1}^{|\cM|} v^\star_i(e_q).
    \]
    Since the construction of $v^\star_i$ preserves the total cost, the right-hand side equals $\sum_{q=1}^{|\cM|} v_i(e_q)$, contradicting the nonnegativity of all costs. Therefore $v^\star_i(e_s)>0$. Let $v^\star_{i}(e_s) = w_p$. Since $v_{i}(e_s) > v^\star_{i}(e_s)$, the definition of canonical instance implies that $v_{i}(e_s) \geq 2w_p$.}\acomment{Author approved adding the cumulative zero-case argument while avoiding new notation. Original text: Let $v^\star_{i}(e_s) = w_p$. Since $v_{i}(e_s) > v^\star_{i}(e_s)$, the definition of canonical instance implies that $v_{i}(e_s) \geq 2w_p$. Original unresolved comment asked to keep the minimal-failing-prefix inequality for the case $v_i^\star(e_s)=0$.}
    Let $\B$ be an arbitrary WMMS partition for $a_{i}$ and let $\overline{B} = \bigcup_{q=1}^{p-1}\bigcup_{r=1}^{n_q}B_{q,r}$. Since $v_{i}(e_h) > w_p, \forall h\leq s, h\in \mathbb{N}_+$, each of the first $s$ items needs to fit in $\overline{B}$. Therefore, we have the following.
    \begin{align*}
        v_{i}(\overline{B}) &\geq  \sum_{q=1}^{d_i} v_{i}(e_q) + \sum_{q=d_i+1}^{s} v_{i}(e_q)
       > \sum_{q=1}^{d_i} v^\star_{i}(e_q) + \sum_{q=d_i+1}^{s} v^\star_{i}(e_q)  \\
        &\geq \sum_{q=1}^{s} v^\star_{i}(e_q) \geq \sum_{q=1}^{\achange{L_i}+n_i}w_i + \sum_{q=\achange{L_i}+n_i+1}^{\achange{L_i}+n_i+n_{i+1}}w_{i+1}
        +\cdots+\sum_{q=\achange{L_i}+\sum_{r\in \{n_i,\ldots, n_{p-2}\}}n_r+1}^{\achange{L_i}+\sum_{r\in \{n_i,\ldots, n_{p-1}\}}n_r}w_{p-1}\\
        &\geq (\frac{\sum_{q<i}W_q }{w_i}+n_i)w_i + n_{i+1}\cdot w_{i+1}+\cdots + n_{p-1}\cdot w_{p-1} \\
        &\geq \sum_{q=1}^{p-1}W_q.
    \end{align*}
    The second inequality in the first line requires that $\sum_{q=1}^{d_i} v_{i}(e_q) \geq \sum_{q=1}^{d_i} v^\star_{i}(e_q)$, which is obvious since for each $1\leq h \leq d_i, v_{i}(e_h) > w_i$ while $v^\star_{i}(e_h) = w_i$.
    This leads to a contradiction because $\B$ being a WMMS partition requires that
    \[
        v_{i}(\overline{B})= \sum_{q=1}^{p-1}\sum_{r=1}^{n_q}v_{i}(B_{q,r}) \leq \sum_{q=1}^{p-1}\sum_{r=1}^{n_q}w_q \leq \sum_{q=1}^{p-1}W_q.
     \]
\end{proof}

Combining Inequality \ref{eq:swh-1:1} and Claim \ref{cla:star}, we have

\[
\begin{aligned}
sc_{h-1}(\cN) 
&\leq v^{\star}_{r_1}([1, d_{r_2}]) + v^{\star}_{r_2}([d_{r_2}+1, d_{r_3}]) + \cdots + v^{\star}_{r_t}([d_{r_t}+1, h-1]) - \sum_{i\in [p-1]} Y_i.
\end{aligned}
\]

It remains to show that
\[
v^{\star}_{r_1}([1, d_{r_2}]) + v^{\star}_{r_2}([d_{r_2}+1, d_{r_3}]) + \cdots + v^{\star}_{r_t}([d_{r_t}+1, h-1]) - \sum_{i\in [p-1]} Y_i
< 3W_1 + \cdots + 3W_p.
\]
\acomment{Follow-up author verification: The strict displayed inequality is now stated, but the proof still needs the bridge from Claims~\ref{cla:y1} and~\ref{cla:yi} to the terminal expression ending at $h-1$. Please indicate where the strict slack enters for this truncated endpoint, rather than only for the full proxy prefix ending at the next $L_{r_{i+1}}$ boundary.}


We begin by assigning each item a label from $\{\gamma, \lambda, \kappa\}$ and associating it with a value determined by one of the valuations $\{v^{\star}_{r_1},\,\ldots,\,v^{\star}_{r_t}\}$.

Define function $v(e)_u = \max\{v(e), u\}$ and $v(S)_u = \sum_{e\in S}v(e)_u$. 
To facilitate bounding the value of each $Y_i$, we employ $v^{\star}_{r_i}(\cdot)_{w_{r_{i+1}}}$ to estimate the value of the set $[{d_{r_i}+1},{d_{r_{i+1}}}]$.

\begin{claim}\label{cla:y1}
For the valuation $v_{r_1}^{\star}(\cdot)$ and the set $[{1},{L_{r_2}}]$, we have
\[
v^{\star}_{r_1}([1, L_{r_2}]) \leq v^{\star}_{r_1}([1, L_{r_2}])_{w_{r_2}} < 3W_1 + \cdots + 3W_{r_2-1}.
\]
\end{claim}




\begin{proof}

{The first inequality is immediate from the definition of $v(\cdot)_u$, so we focus on the second inequality.}

We begin by labeling and assigning costs to the items in $S_1 = [{1}, {L_{r_2}}]$ according to the cost function $v^{\star}_{r_1}(\cdot)$. 
In this proof, each item $e \in [{1}, {L_{r_2}}]$ is assigned cost $v^{\star}_{r_1}(e)$ and one label from $\{\gamma, \lambda, \kappa\}$.
Any item with index greater than $d_{r_2}$ will later be relabeled and revalued using the remaining cost functions $\{v^{\star}_{r_2}(\cdot), \ldots, v^{\star}_{r_t}(\cdot)\}$ in the proof of Claim \ref{cla:yi}. 

For each $i \in \{r_1, \ldots, r_2-1\}$, the items $[{L_i+1}, {L_i+n_i}]$ are labeled $\gamma$ and the set is denoted as $S^{i}_\gamma$.
Since $L_{i+1}+1=\frac{\sum_{j<i+1}W_j}{w_{i+1}}+1> 2(L_{i}+n_i)$, we have $[{L_i+1}, {L_i+n_i}]\cap [{L_{i+1}+1}, L_{i+1}+n_{i+1}]=\emptyset$  and no item is labeled twice.
By Claim \ref{cla:lr}, each item in the set $[{L_i+1}, {L_i+n_i}]$ has a cost at most $w_i$, and therefore the total cost of $[{L_i+1}, {L_i+n_i}]$ does not exceed $W_i$.
Hence, the total cost of the set of items with label $\gamma$ is less than $\sum_{j<r_2} W_j$.

The remaining items in $S_1$ are labeled $\lambda$. 
For $i \in \{r_1+1, \ldots, r_2-1\}$, let $S^i_\lambda$ denote the set of items in $S_1$ with cost $w_i$ and label $\lambda$, and let $X_i$ be the total cost of $S^i_\lambda$.
Let $\bar{X}_i$ be the difference between $X_i$ and $\sum_{j<i}W_j$ for any $i>1$.

{First, consider the total cost of the set $S^2_\lambda$.}
By Claim \ref{cla:lr}, the set $S^2_\lambda$ is a proper subset of $[1, L_2]$, so $|S^2_\lambda| < L_2$ and thus $X_2=|S^2_\lambda|\cdot w_2 < W_1$.
{Since $|S^2_\lambda|$ is also no larger than the number of agents in $G_2$ based on the structure of $v_{r_1}^\star(\cdot)$, it also follows that $X_2 \leq W_2$.}


Similarly, for any $i\geq 2$, we have $|S^{i+1}_\lambda| < L_{i+1}$ by Claim \ref{cla:lr}.
If $X_i < W_i$ and $\bar{X}_i = W_1 + \cdots + W_{i-1} - X_i \geq 0$, then for $S^{i+1}_\lambda$, 
the total cost $X_{i+1}$ is less than $L_{i+1}\cdot w_{i+1}=W_1 + W_2 + \cdots + W_i = \bar{X}_i + X_i + W_i \leq \bar{X}_i + 2W_i$. 
In addition, since $|S^{i+1}_\lambda| \leq n_{i+1}$, we have $X_{i+1} \leq W_{i+1}$.
{By induction, we have $X_{i+1}\leq \bar{X}_{i}+2W_i$ for any $2\leq i< r_2-1$.}

{For the remaining items in $S_1$, each item has cost less than $w_{r_2}$, and the number of such items is less than $L_{r_2}$.} Therefore, the total cost of these items is less than
\[
L_{r_2} \cdot w_{r_2} = W_1 + \cdots + W_{r_2-1} \leq \bar{X}_{r_2-1} + 2W_{r_2-1}.
\]

In summary, we have
\[
\begin{aligned}
    v^{\star}_{r_1}([1, L_{r_2}])_{w_{r_2}}
    < & v^{\star}_{r_1}\left(\bigcup_{i=r_1}^{r_2-1} S^i_\gamma\right)
        + v^{\star}_{r_1}\left(\bigcup_{i=r_1+1}^{r_2-1} S^i_\lambda\right)
        + L_{r_2} \cdot w_{r_2} \\
    < & \sum_{i=1}^{r_2-1} W_i + \sum_{i=2}^{r_2-1} X_i + L_{r_2} \cdot w_{r_2} \\
    \leq& \sum_{i=1}^{r_2-1} W_i + (W_1 - \bar{X}_2) + (W_1 + W_2 - \bar{X}_3) + \cdots \\
    &\quad\quad\quad\quad\quad\quad\quad\quad + (W_1 + \cdots + W_{r_2-2} - \bar{X}_{r_2-1}) + \sum_{i=1}^{r_2-1} W_i \\
    \leq& \sum_{i=1}^{r_2-1} W_i + (W_1 - \bar{X}_2) + (\bar{X}_2+2W_2 - \bar{X}_3) + \cdots  \\
    &\quad\quad\quad\quad\quad\quad\quad\quad  + (\bar{X}_{r_2-2} + 2W_{r_2-2} - \bar{X}_{r_2-1}) + \bar{X}_{r_2-1}+2W_{r_2-1}  \\
    \leq& \sum_{i=1}^{r_2-1} W_i + W_1 + \sum_{i=2}^{r_2-1} 2W_i \leq \sum_{i=1}^{r_2-1} 3W_i.
\end{aligned}
\]


\end{proof}


\begin{claim}\label{cla:yi}
    For the valuation $v_i^{\star}(\cdot)$ and items $[{d_{r_i}},{L_{r_{i+1}}}]$ with $i \in \{2, \ldots, r_t-1\}$, 
    if for any $i' < i$,
    \begin{equation}\label{eq:i'}
    \begin{aligned}
    &\sum_{j=1}^{i'-1} v^{\star}_{r_j}([d_{r_j}+1, d_{r_{j+1}}]) + v^{\star}_{r_{i'}}([d_{r_{i'}}+1, L_{r_{i'+1}}]) - \sum_{j=1}^{i'-1} Y_j \\
    \leq\; &
    \sum_{j=1}^{i'-1} v^{\star}_{r_j}([d_{r_j}+1, d_{r_{j+1}}])_{w_{r_{i'+1}}} + v^{\star}_{r_{i'}}([d_{r_{i'}}+1, L_{r_{i'+1}}])_{w_{r_{i'+1}}} - \sum_{j=1}^{i'-1} Y_j \\
    <\; &
    3W_1 + \cdots + 3W_{r_{i'+1}-1},
    \end{aligned}
    \end{equation}
    then we have 
    \[
    \begin{aligned}
    &\sum_{j=1}^{i-1} v^{\star}_{r_j}([d_{r_j}+1, d_{r_{j+1}}]) + v^{\star}_{r_i}([d_{r_i}+1, L_{r_{i+1}}]) - \sum_{j=1}^{i-1} Y_j \\
    \leq\; &
    \sum_{j=1}^{i-1} v^{\star}_{r_j}([d_{r_j}+1, d_{r_{j+1}}])_{w_{r_{i+1}}} + v^{\star}_{r_i}([d_{r_i}+1, L_{r_{i+1}}])_{w_{r_{i+1}}} - \sum_{j=1}^{i-1} Y_j \\
    <\; &
    3W_1 + \cdots + 3W_{r_{i+1}-1}.
    \end{aligned}
    \]
\end{claim}


    


We prove Claim \ref{cla:yi} in Appendix~\ref{sec:appendix:3WMMS}.

{By Claims~\ref{cla:y1} and~\ref{cla:yi}, together with induction, we obtain}
\[
v^{\star}_{r_1}([1, d_{r_2}]) + v^{\star}_{r_2}([d_{r_2}+1, d_{r_3}]) + \cdots + v^{\star}_{r_t}([d_{r_t}+1, h-1]) - \sum_{i\in [t-1]} Y_i
< 3W_1 + \cdots + 3W_p,
\]
which proves Lemma~\ref{lem:allocated}.

\medskip

Given Claim \ref{claim:exact} and Lemma \ref{lem:allocated}, Theorem~\ref{the:main} is proved.
{Combining this with Theorem~\ref{the:can}, we obtain a $12$-WMMS algorithm for any fair-allocation instance.}

\section{Lower Bounds}
\label{sec:lower}

In this section, we consider two types of lower bounds: the worst-case lower bound and lower bounds for {an arbitrary fixed number of agents}.

\subsection{The Worst-case Lower Bound}

We first complement our previous algorithmic results by proving that for any $0<\epsilon<1$, there is an instance for which no allocation is better than $(2-\epsilon)$-WMMS.

\begin{theorem}\label{the:lb}
    For any $0<\epsilon<1$, there is an instance for which no allocation is better than $(2-\epsilon)$-WMMS.
\end{theorem}

{Due to space limitations}, in the remainder of this section, we only present the construction of the instance $\overline\cI$.
We formally prove the theorem by contradiction in Appendix \ref{sec:appendix:lb}.
Roughly speaking, assuming there exists an allocation $\A$ with an approximation ratio better than $(2-\epsilon)$, we transform $\A$ into a {\em fractional} allocation $\X$ with additional properties and no agent has higher cost.
We show that even with this relaxation, $\X$ cannot be better than $(2-\epsilon)$-WMMS, contradicting the assumption.

\medskip



Given any $0<\epsilon<1$, let $k > \frac{4}{\epsilon^2}$ be a sufficiently large integer and $\Delta=2^k$. 
We construct the following instance $\overline\cI=(\cN,\cM,\fw,\fv)$, where $\cN$ and $\cM$ are respectively partitioned into $k$ disjoint groups, $\cN=G_1\cup \cdots \cup G_k$ and $\cM= \cM_1\cup \cdots \cup \cM_k$.
Denote by $n_i$ and $m_i$ the numbers of agents and items in {$G_i$} and $\cM_i$, respectively.

For $1 \le i \le k$, let $n_i=\Delta^{i-1}$, and $m_1= 1$ and $m_i= \frac{3}{2}\cdot n_i = \frac{3}{2} \cdot \Delta^{i-1}$ for $2 \le i \le k$.
Thus, we have, for $2\le i \le k-1$, 
\begin{align}
    \label{eq:lower:instance:m}
    m_1 + \cdots + m_{i-1} \le \frac{3}{2} \cdot \frac{\Delta^{i-1}}{\Delta - 1} < 2 \cdot \Delta^{i-2} \le \frac{1}{6} \cdot \Delta^{i-1}  = \frac{1}{6}\cdot n_i = \frac{1}{9} \cdot m_i.
\end{align}
Agents in the same group have the same weight, and let $w_1 = 1$ and $w_i = \frac{1}{n_i}\cdot 2^{i-2}$ for every agent in group $G_i$ with $i \ge 2$.
Thus, the total weight of agents in $G_i$ is $W_i= n_i \cdot w_i = 2^{i-2}$ for $i \ge 2$.
The agents in the same group have the same valuation function. 
For simplicity, we use $v_i(\cdot)$ and $\WMMS_i$ to represent the valuation function and WMMS of every agent in $G_i$.

\begin{figure}[h]
        \centering
        \includegraphics[width=0.9\linewidth]{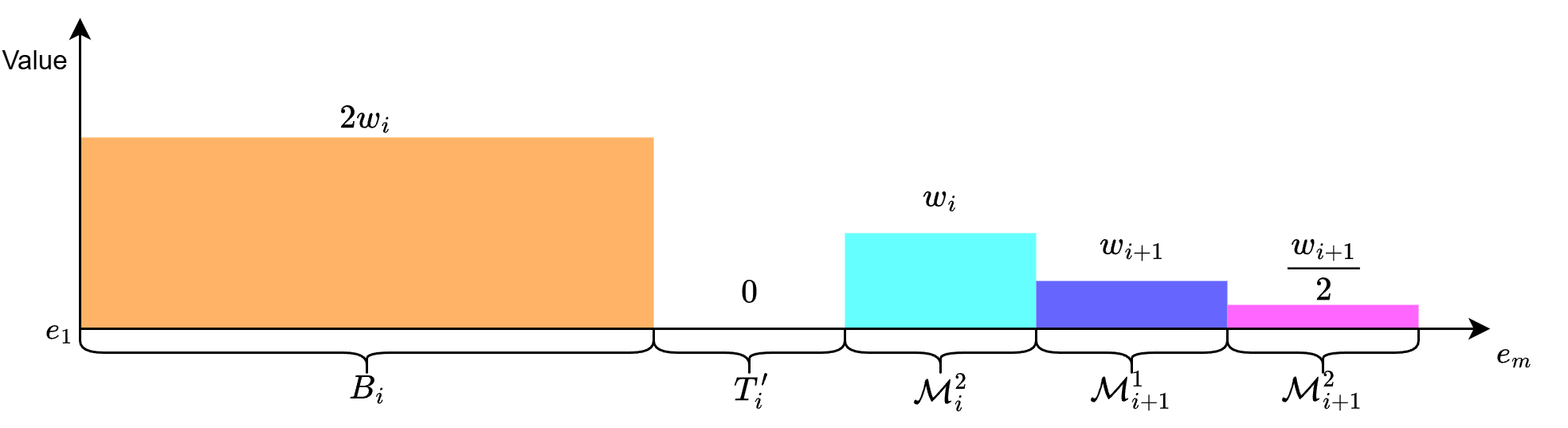}
        \caption{An illustration of $v_i(\cdot)$.}
        \label{fig:enter-label}
\end{figure}
    
Next, we design the valuations, with an illustration of $v_i(\cdot)$ in Figure \ref{fig:enter-label}.
\begin{itemize}
    \item Let $v_1(e) = 1$ for the only item $e \in \cM_1$.

    \item For any $1\le j < i$ and $i\ge 2$, (arbitrarily) partition $\cM_i$ into $\cM_i= \cM_i^1 \cup \cM_i^2$ such that $|\cM_i^1| = \frac{1}{3} m_i = \frac{1}{2}n_i$ and $|\cM_i^2| = \frac{2}{3} m_i = n_i$.
    Let 
\begin{align*}
        v_j(e) = \begin{cases}
        \frac{2^{i-2}}{n_i} = w_i, & \text{ for any $e\in \cM_i^1$} \\
        \frac{2^{i-2}}{2n_i} = \frac{1}{2}w_i, & \text{ for any $e\in \cM_i^2$}.
\end{cases}
\end{align*}
\item For $i>1$,
let $T_i$ be an arbitrary subset of $\cM_i^1$ with $|T_i| = \frac{1}{3} m_i - (m_1+\cdots+m_{i-1})$ and let $T_i' = \cM_i^1 \setminus T_i$. 
Note that $T_i$ and $T'_i$ are well defined because of Inequality \ref{eq:lower:instance:m}.
Denote by $B_i = \cM_1\cup \cdots \cup \cM_{i-1} \cup T_i$ for $i> 1$ and $B_1 = \emptyset$. 
Note that 
    \[
    |B_i| = \frac{1}{2}n_i = \frac{1}{3} m_i \gg |\cM_1\cup \cdots \cup \cM_{i-1}|.
    \]
Set $v_i(\cdot)$ for $\cM_1\cup \cdots \cup \cM_i$ with $i>1$ as:
    \begin{align*}
        v_i(e) = \begin{cases}
        \frac{2^{i-1}}{n_i} = 2w_i, & \text{  for any $e\in B_i$} \\
        0, & \text{ for any $e\in T_i'$} \\
        \frac{2^{i-2}}{2n_i} = w_i, & \text{  for any $e\in \cM_i^2$}.
    \end{cases}
    \end{align*}
\end{itemize}

Next, we prove that each agent's weighted maximin share is exactly her weight. 
\begin{lemma}
    For any group of agents $G_i$, $\WMMS_i = w_i$.
\end{lemma}

\begin{proof}
    For the sole agent in $G_1$, one WMMS-defining partition is as follows:
    \begin{itemize}
        \item The agent in $G_1$ gets the sole item in $\cM_1$;
        \item Each agent $a_{i,j} \in G_i$ for $i\ge 2$ gets either exactly one item in $\cM_i^1$ or exactly two items in $\cM_i^2$.
    \end{itemize}
    {In this partition, all items are allocated.}
    Besides, from the perspective of the agent in $G_1$, every agent in $\cN$ is allocated a set of items whose value is exactly her weight; therefore, this partition is WMMS-defining.
    Thus, $\WMMS_1 = w_1$.

    Next, consider agents in group $G_i$ with $i \ge 2$. 
    We first observe that 
    \[
    \sum_{j =1}^{i-1} W_j= 1+\sum_{j =0}^{i-3} 2^{j} = 2^{i-2} =
    \left( 2\cdot \frac{2^{i-2}}{n_i}\right) \cdot \left(\frac{1}{2} n_i \right)= 2w_i \cdot |B_i| = v_i(B_i).
    \]
    That is, from the perspective of agents in $G_i$, all the items in $B_i$ can be allocated to agents in $G_1\cup \cdots\cup G_{i-1}$ such that every agent gets a bundle of items whose value equals her weight. This is possible since the value of any item in $B_i$ is $2w_i$, which divides each $w_j$ with $j<i$. 
    Thus, one WMMS partition for agents in $G_i$ can be constructed as follows:
    \begin{itemize}
        \item The items in $B_i$ are allocated to agents in $G_1\cup \cdots\cup G_{i-1}$ such that each agent gets a cost {equal to} her weight;
        \item Each agent in $G_i$ is {allocated one item} in $\cM_i^2$;
        \item The items in $T_i'$ are arbitrarily allocated to the agents in $G_i$;
        \item Each agent in $G_j$ with $j > i$ gets either exactly one item in $\cM_j^1$ or exactly two items in $\cM_j^2$.
    \end{itemize}
    One can easily verify that every agent in $\cN$ is allocated a bundle with cost equal to her weight, and all the items are allocated.
    Therefore, $\WMMS_i = w_i$.
\end{proof}


\subsection{Lower Bounds for {an Arbitrary Number of Agents}}

It is worth noting that the previous lower bound of 2 applies only when $n$ is sufficiently large.  
While the lower bound for $n=2$ has been proved in \citep{DBLP:conf/ijcai/00020L24}, the lower bounds for intermediate values of $n$ remain unknown.  
To address this gap, we present a more general instance below. Although it is weaker than Theorem \ref{the:lb} in the worst case, it provides a lower bound for any number of agents.

\begin{theorem}
    \label{thm:lb:n}
    For any number of agents $n\ge 2$, there exists an instance in which no allocation achieves better than $\frac{n+\sqrt{5n^2-4n}}{2n}$-WMMS.
\end{theorem}

The proof of Theorem \ref{thm:lb:n} is given in Appendix \ref{sec:appendix:lb}. 

For illustration, when $n=2$, no allocation is better than 1.366-WMMS, which matches the lower bound in \citep{DBLP:conf/ijcai/00020L24}; when $n=3$, no allocation is better than 1.457-WMMS; {when} $n=4$, the lower bound is 1.5, and as $n$ approaches infinity, the lower bound converges to the golden ratio, approximately 1.618.


\section{WMMS Approximation Curve for Two Agents with Varying Weights}

\label{sec:two}

In this section, we provide a deeper discussion on how weights affect the approximation ratio of WMMS.
It is well-known that, for the case of two agents with equal weights, an exact MMS allocation always exists, which can be found by, for example, the cut-and-choose algorithm. 
As shown in \citep{DBLP:conf/ijcai/00020L24}, when the two weights are $w_1 = \sqrt{3}-1$ and $w_2 = 2-\sqrt{3}$,  there exist valuation profiles for which no allocation can be better than $\frac{\sqrt{3}+1}{2}$-WMMS. 
This is in fact the worst-case weight distribution; for any weight distribution {and any valuation profiles}, a $\frac{\sqrt{3}+1}{2}$-WMMS allocation exists.
This leads to a natural question: if the two agents have a different weight distribution, {to what extent can the approximation ratio be improved?}
To address this, in what follows, we provide a precise characterization of the optimal approximation ratio as a function of the weight distribution for the two-agent case, as shown in Theorem \ref{thm:n=2:curve}.

For simplicity, assume $w_1\ge w_2>0$ throughout this section.
Denote by $\Sigma(\alpha)$ the set of all two-agent instances $\cI$ with $\alpha=\frac{w_1}{w_2}\ge 1$.
Define $f(\alpha)$ as the worst-case optimal approximation ratio of instances in $\Sigma(\alpha)$, i.e.,
\[
f(\alpha)=\max_{\cI\in \Sigma(\alpha)}\min_{\rho}\{\rho\mid \text{$\exists$ allocation $(X_1,X_2)$ of $\cI$ s.t. $v_i(X_i)\le \rho\WMMS_i$ for $i=1,2$}\}.
\]

\begin{theorem}
\label{thm:n=2:curve}
{For every $\alpha\ge 1$, the worst-case optimal approximation ratio over all instances in $\Sigma(\alpha)$ is}
    \begin{equation*}  
f(\alpha)=\left\{  
     \begin{array}{ll}    
     \frac{1}{2}+\frac{\alpha}{2}, & 1\leq \alpha < \frac{\sqrt{5}+1}{2}\\
      1+\frac{1}{2\alpha},  & \frac{\sqrt{5}+1}{2}\leq \alpha < \sqrt{2}+1\\
     \frac{\alpha}{2}, & \sqrt{2}+1\le \alpha < \sqrt{3}+1\\
      1+\frac{1}{\alpha}, & \alpha \ge \sqrt{3}+1.
     \end{array}  
\right.  
\end{equation*} 
 \end{theorem}

The function $f(\alpha)$ is illustrated in Figure \ref{fig:n=2}.
Theorem \ref{thm:n=2:curve} directly implies that, for any two-agent instance with unequal weights, there exist valuation profiles where an exact WMMS allocation does not exist.

\begin{figure}
    \centering
    \includegraphics[width=0.8\linewidth]{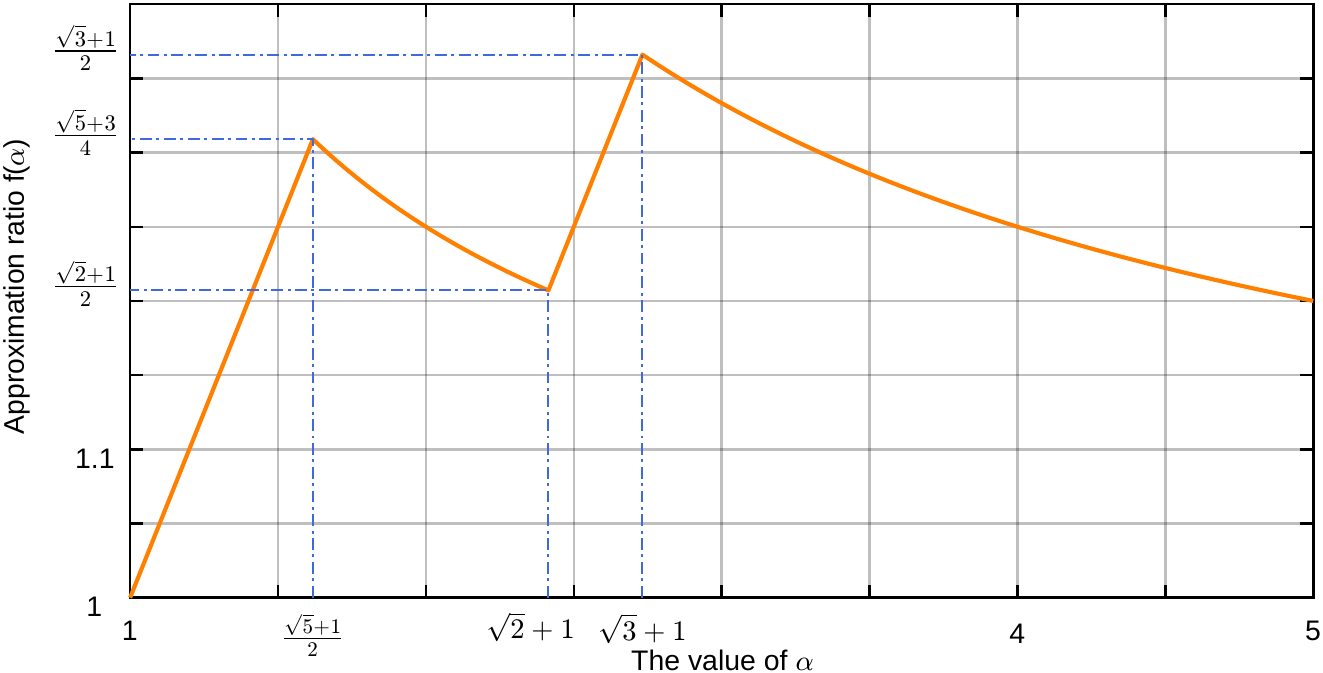}
    \caption{The curve of optimal approximation ratio $f(\alpha)$.}
    \label{fig:n=2}
\end{figure}

\begin{corollary}
    Given any two agents with $w_1>w_2$, there are valuations $(v_1,v_2)$ such that no allocation is WMMS.
\end{corollary}

We divide the proof of Theorem \ref{thm:n=2:curve} into two lemmas.

\begin{lemma}
\label{lem:two:lower}
    For any $\alpha\ge 1$, there exists a two-agent instance $\cI=(\cN,\cM,\fw,\fv)$ with $\frac{w_1}{w_2}=\alpha$ such that for any allocation $(X_1,X_2)$, 
    \[
    \max\{\frac{v_1(X_1)}{\WMMS_1}, \frac{v_2(X_2)}{\WMMS_2}\} \ge f(\alpha).
    \]
\end{lemma}

\begin{proof}
To prove the lemma, we consider two hard instances, depending on the weights $w_1$ and $w_2$:
instance $\cI_1$ when $1\le\alpha\le \sqrt{2}+1$, and instance $\cI_2$ otherwise. 
    
\paragraph{Instance $\cI_1=(\cN,\cM,\fw,\fv)$}
In $\cI_1$, $\cN=\{a_1,a_2\}$, $\cM=\{e_1,e_2,e_3\}$, and the valuations are shown in Table \ref{tab:valuation1}.
{A direct calculation gives $\WMMS_1=w_1$ and $\WMMS_2=w_2$.}

\begin{table}[htbp]
    \centering
    \renewcommand{\arraystretch}{1.2}
    \begin{minipage}[b]{55mm}
    \begin{tabular}{c|c|c|c}
        \hline
        agent  & item $e_1$ & item $e_2$ & item $e_3$\\
        \hline
        $a_1$ & $w_1$ & $\frac{w_2}{2}$ & $\frac{w_2}{2}$ \\
        \hline
        $a_2$ & $\frac{w_1+w_2}{2}$ & $w_2$ & $\frac{w_1-w_2}{2}$ \\
        \hline
    \end{tabular}
    \caption{$\cI_1=(\cN,\cM,\fw,\fv)$.}
    \label{tab:valuation1}
    \end{minipage}  
    \begin{minipage}[b]{55mm}
    \begin{tabular}{c|c|c|c}
        \hline
        agent  & item $e_1$ & item $e_2$ & item $e_3$\\
        \hline
        $a_1$ & $w_1$ & $0$ & $w_2$ \\
        \hline
        $a_2$ & $\frac{w_1}{2}$ & $w_2$ & $\frac{w_1}{2}$ \\
        \hline
    \end{tabular}
    \caption{$\cI_2=(\cN,\cM,\fw,\fv)$.}
    \label{tab:valuation2}
    \end{minipage}
\end{table}

To see the best possible approximation of WMMS in $\cI_1$, we enumerate all its possible allocations as shown in Table \ref{tab:profile1}.
For each allocation, the unfairness ratio is given by ${\max\{\frac{v_1(A_1)}{w_1},\frac{v_2(A_2)}{w_2}\}}$, and the agent whose allocation touches this ratio is marked in green.
Given $1\le \alpha\le \sqrt{2}+1$, the 6th and 7th allocations dominate the 8th allocation, and the 4th and 5th allocations dominate the 1st, 2nd, and 3rd allocations.
{Here, one allocation dominates another if it has a weakly smaller unfairness ratio for every relevant value of $\alpha$.}
Thus, on instance $\cI_1$, the best possible approximation ratio is $\min\{1+\frac{w_2}{2w_1},\frac{1}{2}+\frac{w_1}{2w_2}\}$, or alternatively, 
\begin{equation*}  
L_1(\alpha)=\left\{  
    \begin{array}{cc}
     \frac{1}{2}+\frac{\alpha}{2},& 1\le \alpha<\frac{\sqrt{5}+1}{2} \\
     1+\frac{1}{2\alpha}, & \frac{\sqrt{5}+1}{2} \le \alpha\le \sqrt{2}+1. \\
    \end{array}
\right.  
\end{equation*} 

\begin{table}[htbp]
  \centering
  
    \begin{tabular}{c|c|c|c|c}
    \hline
    ID & $A_1$ & $A_2$ & $v_1(A_1)$ & $v_2(A_2)$ \bigstrut\\
    \hline
    1     & $\emptyset$ & $\achange{\{e_1,e_2,e_3\}}$ & $0$   & \cellcolor[rgb]{ .757,  .941,  .784} $w_1+w_2$ \\
    \hline
    2     & $\achange{\{e_3\}}$ & $\achange{\{e_1,e_2\}}$ & $0.5w_2$ & \cellcolor[rgb]{ .757,  .941,  .784} $0.5w_1+1.5w_2$ \\
    \hline
    3     & $\achange{\{e_2\}}$ & $\achange{\{e_1,e_3\}}$ & $0.5w_2$ & \cellcolor[rgb]{ .757,  .941,  .784} $w_1$ \\
    \hline
    4     & $\achange{\{e_2,e_3\}}$ & $\achange{\{e_1\}}$ & $w_2$ & \cellcolor[rgb]{ .757,  .941,  .784} $0.5w_1+0.5w_2$ \\
    \hline
    5     & $\achange{\{e_1\}}$ & $\achange{\{e_2,e_3\}}$ & $w_1$ & \cellcolor[rgb]{ .757,  .941,  .784} $0.5w_1+0.5w_2$ \\
    \hline
    6     & $\achange{\{e_1,e_3\}}$ & $\achange{\{e_2\}}$ & \cellcolor[rgb]{ .757,  .941,  .784} $w_1+0.5w_2$ & $w_2$ \bigstrut\\
    \hline
    7     & $\achange{\{e_1,e_2\}}$ & $\achange{\{e_3\}}$ & \cellcolor[rgb]{ .757,  .941,  .784} $w_1+0.5w_2$ & $0.5w_1-0.5w_2$ \bigstrut\\
    \hline
    8     & $\achange{\{e_1,e_2,e_3\}}$ & $\emptyset$ & \cellcolor[rgb]{ .757,  .941,  .784} $w_1+w_2$ & $0$ \bigstrut\\
    \hline
    \end{tabular}%
    \caption{All possible allocations in $\cI_1$.}
  \label{tab:profile1}%
\end{table}%

When $\alpha>\sqrt{2}+1$, we consider the following instance $\cI_2$.

\paragraph{Instance $\cI_2=(\cN,\cM,\fw,\fv)$}
Again, in $\cI_2$, $\cN=\{a_1,a_2\}$ and $\cM=\{e_1,e_2,e_3\}$; the valuations are shown in Table \ref{tab:valuation2}, and $\WMMS_1=w_1$, $\WMMS_2=w_2$.

\begin{table}[htbp]
  \centering
  
    \begin{tabular}{c|c|c|c|c}
    \hline
    ID & $A_1$ & $A_2$ & $v_1(A_1)$ & $v_2(A_2)$ \bigstrut\\
    \hline
    1     & $\emptyset$ & $\achange{\{e_1,e_2,e_3\}}$ & $0$   & \cellcolor[rgb]{ .757,  .941,  .784} $w_1+w_2$ \\
    \hline
    2     & $\achange{\{e_3\}}$ & $\achange{\{e_1,e_2\}}$ & $w_2$ & \cellcolor[rgb]{ .757,  .941,  .784} $0.5w_1+w_2$ \\
    \hline
    3     & $\achange{\{e_2\}}$ & $\achange{\{e_1,e_3\}}$ & $0$   & \cellcolor[rgb]{ .757,  .941,  .784} $w_1$ \\
    \hline
    4     & $\achange{\{e_2,e_3\}}$ & $\achange{\{e_1\}}$ & $w_2$ & \cellcolor[rgb]{  .757,  .941,  .784} $0.5w_1$ \\
    \hline
    5     & $\achange{\{e_1\}}$ & $\achange{\{e_2,e_3\}}$ & $w_1$ & \cellcolor[rgb]{ .757,  .941,  .784} $0.5w_1+w_2$ \\
    \hline
    6     & $\achange{\{e_1,e_3\}}$ & $\achange{\{e_2\}}$ & \cellcolor[rgb]{ .757,  .941,  .784} $w_1+w_2$ & $w_2$ \bigstrut\\
    \hline
    7     & $\achange{\{e_1,e_2\}}$ & $\achange{\{e_3\}}$ & $w_1$ & \cellcolor[rgb]{  .757,  .941,  .784} $0.5w_1$ \\
    \hline
    8     & $\achange{\{e_1,e_2,e_3\}}$ & $\emptyset$ & \cellcolor[rgb]{ .757,  .941,  .784} $w_1+w_2$ & $0$ \bigstrut\\
    \hline
    \end{tabular}%
    \caption{All possible allocations in $\cI_2$.}
  \label{tab:profile2}%
\end{table}%

Similar to the analysis for $\cI_1$, we enumerate all possible allocations in $\cI_2$, as shown in Table \ref{tab:profile2}, and color the agent who reaches the unfairness ratio.
Since $\alpha>\sqrt{2}+1$, allocations $1,2,3,5$ are dominated by allocations $4$ and $7$.
Accordingly, the best possible approximation is $\min\{1+\frac{w_2}{w_1},\frac{w_1}{2w_2}\}$, or alternatively, 
\begin{equation*}  
L_2(\alpha)=\left\{  
    \begin{array}{cc}
       \frac{\alpha}{2}, & \sqrt{2}+1<\alpha\le\sqrt{3}+1 \\
        1+\frac{1}{\alpha}, & \alpha>\sqrt{3}+1. \\
    \end{array}
\right.  
\end{equation*} 
Combining $L_1(\alpha)$ and $L_2(\alpha)$, Lemma \ref{lem:two:lower} is proved.
\end{proof}

\begin{lemma}
\label{lem:two:upper}
    For any $\alpha\ge 1$ and any two-agent instance $\cI=(\cN,\cM,\fw,\fv)$ with $\frac{w_1}{w_2}=\alpha$, there exists an allocation $(X_1,X_2)$, 
    \[
    \max\{\frac{v_1(X_1)}{\WMMS_1}, \frac{v_2(X_2)}{\WMMS_2}\} \le f(\alpha).
    \]
\end{lemma}

Due to space constraints, we provide only a proof sketch of Lemma \ref{lem:two:upper} below, with the full proof deferred to Appendix~\ref{sec:appendix:two}.

\begin{proof}[Proof Sketch]
    Given an instance with two agents $\{a_1,a_2\}$, denote by $\{B_1,B_2\}$ the WMMS-defining partition of agent $a_1$ and $\{C_1,C_2\}$ the WMMS-defining partition of agent $a_2$. 
    By the proportional proposition in Theorem \ref{the:can}, without loss of generality, it is assumed that $\WMMS_i=w_i$, and $v_1(B_1)=v_2(C_1)=w_1$ and $v_1(B_2)=v_2(C_2)=w_2$.
    We divide all items into 4 atomic bundles:
$X_{1,1}=B_1\cap C_1$, $X_{1,2}=B_1\cap C_2$, $X_{2,1}=B_2\cap C_1$, $X_{2,2}=B_2\cap C_2$.
Accordingly, there are coefficients $0\le a,b,c,d\le 1$ such that the values of the atomic bundles can be {represented} as in Table \ref{tab:parts:1}.

\begin{table}[htbp]
    \centering
    \begin{tabular}{c|c|c|c|c}
        \hline
        agents  & $X_{1,1}=B_1\cap C_1$ & $X_{1,2}=B_1\cap C_2$ & $X_{2,1}=B_2\cap C_1$ & $X_{2,2}=B_2\cap C_2$ \\
        \hline
        $a_1$ & $aw_1$ & $(1-a)w_1$ & $dw_2$ & $(1-d)w_2$\\
        \hline
        $a_2$ & $bw_1$ & $(1-c)w_2$ & $(1-b)w_1$ & $cw_2$\\
        \hline
    \end{tabular}
    \caption{Values for {atomic bundles}.}
    \label{tab:parts:1}
\end{table}

Our algorithm is simple: 
\begin{itemize}
    \item It considers all allocations where the four atomic bundles are allocated whole, without being divided, and chooses the allocation that minimizes the unfairness ratio. 
\end{itemize}
It is somewhat surprising that this simple algorithm achieves the optimal WMMS approximation ratio for any weight distribution.
In Appendix~\ref{sec:appendix:two}, we use a linear programming approach to analyze the worst-case approximation ratio of the above algorithm with respect to the parameters $a, b, c, d$. 
This analysis shows that the approximation ratio is exactly $f(\alpha)$.
\end{proof}

\section{Conclusion and Discussion}

In this paper, we give the first {constant-factor approximation algorithm} for WMMS in indivisible chores
and prove that the best possible approximation ratio is between 2 and 12.
Closing this gap remains an intriguing open problem. 
Our algorithm allows an agent to start receiving items only when the item's value does not exceed her WMMS. 
One possible direction for improvement is to permit agents to receive items even earlier, even if the item's value is larger, potentially achieving a better balance between the value of individual items and the exit threshold. 
However, implementing this approach would require extending the analysis to more general instances beyond {canonical ones}.
The intuition behind this limitation is that, in canonical instances, any item with a value exceeding an agent's WMMS is already at least twice as large as her WMMS. 
Allocating two such items to an agent would result in an approximation ratio worse than what our algorithm achieves.

There are other interesting future research directions as well. 

For example, the $O(\log n)$-approximate algorithm in \cite{DBLP:conf/ijcai/00020L24} is ordinal.
That is, their algorithm does not require the exact numerical values of the items. 
Instead, it relies solely on each agent's ranking of the items by value. 
Ordinal algorithms are simple and easy to implement in practice.
In contrast, our algorithm is not ordinal and requires access to the precise values of all items.
It is unknown {whether} the $O(\log n)$ approximation ratio can be improved by ordinal algorithms.
This problem has been studied in the unweighted setting. 
With ordinal valuations, \citet{DBLP:conf/ijcai/AmanatidisBM16} designed an $\Omega(\frac{1}{\log n})$-approximate MMS algorithm under ordinal valuations, which has been shown to be asymptotically optimal \cite{DBLP:conf/ijcai/0002021}.
For chores, {constant-factor approximation algorithms for MMS} are given in \cite{DBLP:conf/sigecom/FeigeH23}.

It is also interesting to consider the online setting, where items arrive dynamically over time. In the unweighted case, \citet{DBLP:conf/icml/0002B023} presented a $(2-\frac{1}{n})$-competitive algorithm under the total value assumption. {Without this assumption, \citet{DBLP:journals/corr/abs-2507-14039} and \citet{DBLP:journals/corr/abs-2507-12984} show that no algorithm can achieve a competitive ratio better than $n$.}
When agents have arbitrary weights, \citet{DBLP:conf/ijcai/00020L24} gave an algorithm that is $O(\sqrt{n})$-competitive under the total value assumption. However, whether this ratio can be further improved remains an open question.

\newpage

\bibliographystyle{plainnat}
\bibliography{ref}

@article{DBLP:journals/ipl/Suksompong25,
  author       = {Warut Suksompong},
  title        = {Weighted fair division of indivisible items: {A} review},
  journal      = {Inf. Process. Lett.},
  volume       = {187},
  pages        = {106519},
  year         = {2025}
}

@inproceedings{DBLP:conf/icml/0002B023,
  author       = {Shengwei Zhou and
                  Rufan Bai and
                  Xiaowei Wu},
  title        = {Multi-agent Online Scheduling: {MMS} Allocations for Indivisible Items},
  booktitle    = {{ICML}},
  series       = {Proceedings of Machine Learning Research},
  volume       = {202},
  pages        = {42506--42516},
  publisher    = {{PMLR}},
  year         = {2023}
}

@inproceedings{DBLP:conf/ijcai/0002021,
  author       = {Daniel Halpern and
                  Nisarg Shah},
  title        = {Fair and Efficient Resource Allocation with Partial Information},
  booktitle    = {{IJCAI}},
  pages        = {224--230},
  publisher    = {ijcai.org},
  year         = {2021}
}

@inproceedings{DBLP:conf/ijcai/AmanatidisBM16,
  author       = {Georgios Amanatidis and
                  Georgios Birmpas and
                  Evangelos Markakis},
  title        = {On Truthful Mechanisms for Maximin Share Allocations},
  booktitle    = {{IJCAI}},
  pages        = {31--37},
  publisher    = {{IJCAI/AAAI} Press},
  year         = {2016}
}

@article{DBLP:conf/sigecom/BabaioffEF21,
  author       = {Moshe Babaioff and
                  Tomer Ezra and
                  Uriel Feige},
  title        = {Fair-Share Allocations for Agents with Arbitrary Entitlements},
  journal      = {Math. Oper. Res.},
  volume       = {49},
  number       = {4},
  pages        = {2180--2211},
  year         = {2024}
}

@article{DBLP:journals/teco/BarmanK20,
  author       = {Siddharth Barman and
                  Sanath Kumar Krishnamurthy},
  title        = {Approximation Algorithms for Maximin Fair Division},
  journal      = {{ACM} Trans. Economics and Comput.},
  volume       = {8},
  number       = {1},
  pages        = {5:1--5:28},
  year         = {2020}
}

@article{Steinhaus48,
	author = {Steinhaus, Hugo},
	title = {The Problem of Fair Division},
	journal = {Econometrica},
	volume = {16},
	number = {1},
	year = {1948},
	pages = {101--104},
}

@article{DBLP:journals/jair/FarhadiGHLPSSY19,
  author       = {Alireza Farhadi and
                  Mohammad Ghodsi and
                  Mohammad Taghi Hajiaghayi and
                  S{\'{e}}bastien Lahaie and
                  David M. Pennock and
                  Masoud Seddighin and
                  Saeed Seddighin and
                  Hadi Yami},
  title        = {Fair Allocation of Indivisible Goods to Asymmetric Agents},
  journal      = {J. Artif. Intell. Res.},
  volume       = {64},
  pages        = {1--20},
  year         = {2019}
}

@article{DBLP:journals/corr/abs-2507-12984,
  author       = {Masoud Seddighin and
                  Saeed Seddighin},
  title        = {Lower Bound for Online {MMS} Assignment of Indivisible Chores},
  journal      = {CoRR},
  volume       = {abs/2507.12984},
  year         = {2025}
}

@article{DBLP:journals/corr/abs-2507-14039,
  author       = {Jiaxin Song and
                  Biaoshuai Tao and
                  Wenqian Wang and
                  Yuhao Zhang},
  title        = {Online {MMS} Allocation for Chores},
  journal      = {CoRR},
  volume       = {abs/2507.14039},
  year         = {2025}
}

@inproceedings{DBLP:conf/ijcai/AmanatidisBFV22,
  author       = {Georgios Amanatidis and
                  Georgios Birmpas and
                  Aris Filos{-}Ratsikas and
                  Alexandros A. Voudouris},
  title        = {Fair Division of Indivisible Goods: {A} Survey},
  booktitle    = {{IJCAI}},
  pages        = {5385--5393},
  publisher    = {ijcai.org},
  year         = {2022}
}

@inproceedings{DBLP:journals/corr/abs-2305-16081,
  author       = {Max Springer and
                  MohammadTaghi Hajiaghayi and
                  Hadi Yami},
  title        = {Almost Envy-Free Allocations of Indivisible Goods or Chores with Entitlements},
  booktitle    = {{AAAI}},
  pages        = {9901--9908},
  publisher    = {{AAAI} Press},
  year         = {2024}
}

@inproceedings{DBLP:conf/sigecom/0001Z023,
  author       = {Xiaowei Wu and
                  Cong Zhang and
                  Shengwei Zhou},
  title        = {Weighted {EF1} Allocations for Indivisible Chores},
  booktitle    = {{EC}},
  pages        = {1155},
  publisher    = {{ACM}},
  year         = {2023}
}

@article{DBLP:journals/teco/CaragiannisKMPS19,
  author       = {Ioannis Caragiannis and
                  David Kurokawa and
                  Herv{\'{e}} Moulin and
                  Ariel D. Procaccia and
                  Nisarg Shah and
                  Junxing Wang},
  title        = {The Unreasonable Fairness of Maximum Nash Welfare},
  journal      = {{ACM} Trans. Economics and Comput.},
  volume       = {7},
  number       = {3},
  pages        = {12:1--12:32},
  year         = {2019}
}

@article{DBLP:journals/orl/AzizMS20,
  author       = {Haris Aziz and
                  Herv{\'{e}} Moulin and
                  Fedor Sandomirskiy},
  title        = {A polynomial-time algorithm for computing a Pareto optimal and almost
                  proportional allocation},
  journal      = {Oper. Res. Lett.},
  volume       = {48},
  number       = {5},
  pages        = {573--578},
  year         = {2020}
}

@inproceedings{DBLP:conf/sigecom/ConitzerF017,
  author       = {Vincent Conitzer and
                  Rupert Freeman and
                  Nisarg Shah},
  title        = {Fair Public Decision Making},
  booktitle    = {{EC}},
  pages        = {629--646},
  publisher    = {{ACM}},
  year         = {2017}
}

@inproceedings{DBLP:conf/ecai/GourvesMT14,
  author       = {Laurent Gourv{\`{e}}s and
                  J{\'{e}}r{\^{o}}me Monnot and
                  Lydia Tlilane},
  title        = {Near Fairness in Matroids},
  booktitle    = {{ECAI}},
  series       = {Frontiers in Artificial Intelligence and Applications},
  volume       = {263},
  pages        = {393--398},
  publisher    = {{IOS} Press},
  year         = {2014}
}

@article{varian1973equity,
  title={Equity, envy, and efficiency},
  author={Varian, Hal R},
  year={1973},
  publisher={[Cambridge, MIT]}
}

@article{Foley67,
	author = {Foley, Duncan Karl},
	title = {Resource Allocation and the Public Sector},
	journal = {Yale Economics Essays},
	volume = {7},
	number = {1},
	year = {1967},
	pages = {45--98},
}

@inproceedings{LiptonMaMo04,
  author       = {Richard J. Lipton and
                  Evangelos Markakis and
                  Elchanan Mossel and
                  Amin Saberi},
  title        = {On approximately fair allocations of indivisible goods},
  booktitle    = {{EC}},
  pages        = {125--131},
  publisher    = {{ACM}},
  year         = {2004}
}

@inproceedings{DBLP:conf/www/0037L022,
  author       = {Bo Li and
                  Yingkai Li and
                  Xiaowei Wu},
  title        = {Almost (Weighted) Proportional Allocations for Indivisible Chores},
  booktitle    = {{WWW}},
  pages        = {122--131},
  publisher    = {{ACM}},
  year         = {2022}
}

@inproceedings{DBLP:conf/sigecom/HuangS23,
  author       = {Xin Huang and
                  Erel Segal{-}Halevi},
  title        = {A Reduction from Chores Allocation to Job Scheduling},
  booktitle    = {{EC}},
  pages        = {908},
  publisher    = {{ACM}},
  year         = {2023}
}

@inproceedings{DBLP:conf/sigecom/FeigeH23,
  author       = {Uriel Feige and
                  Xin Huang},
  title        = {On picking sequences for chores},
  booktitle    = {{EC}},
  pages        = {626--655},
  publisher    = {{ACM}},
  year         = {2023}
}

@inproceedings{DBLP:conf/bqgt/Budish10,
  author       = {Eric Budish},
  title        = {The combinatorial assignment problem: approximate competitive equilibrium
                  from equal incomes},
  booktitle    = {{BQGT}},
  pages        = {74:1},
  publisher    = {{ACM}},
  year         = {2010}
}

@inproceedings{DBLP:conf/ijcai/0001C019,
  author       = {Haris Aziz and
                  Hau Chan and
                  Bo Li},
  title        = {Weighted Maxmin Fair Share Allocation of Indivisible Chores},
  booktitle    = {{IJCAI}},
  pages        = {46--52},
  publisher    = {ijcai.org},
  year         = {2019}
}

@article{DBLP:journals/aamas/BouveretL16,
  author       = {Sylvain Bouveret and
                  Michel Lema{\^{\i}}tre},
  title        = {Characterizing conflicts in fair division of indivisible goods using
                  a scale of criteria},
  journal      = {Auton. Agents Multi Agent Syst.},
  volume       = {30},
  number       = {2},
  pages        = {259--290},
  year         = {2016}
}

@inproceedings{DBLP:conf/sigecom/HuangL21,
  author       = {Xin Huang and
                  Pinyan Lu},
  title        = {An Algorithmic Framework for Approximating Maximin Share Allocation
                  of Chores},
  booktitle    = {{EC}},
  pages        = {630--631},
  publisher    = {{ACM}},
  year         = {2021}
}

@inproceedings{DBLP:conf/sigecom/GhodsiHSSY18,
  author       = {Mohammad Ghodsi and
                  Mohammad Taghi Hajiaghayi and
                  Masoud Seddighin and
                  Saeed Seddighin and
                  Hadi Yami},
  title        = {Fair Allocation of Indivisible Goods: Improvements and Generalizations},
  booktitle    = {{EC}},
  pages        = {539--556},
  publisher    = {{ACM}},
  year         = {2018}
}

@inproceedings{DBLP:conf/nips/0037WZ23,
  author       = {Bo Li and
                  Fangxiao Wang and
                  Yu Zhou},
  title        = {Fair Allocation of Indivisible Chores: Beyond Additive Costs},
  booktitle    = {NeurIPS},
  year         = {2023}
}

@article{DBLP:journals/jacm/KurokawaPW18,
  author       = {David Kurokawa and
                  Ariel D. Procaccia and
                  Junxing Wang},
  title        = {Fair Enough: Guaranteeing Approximate Maximin Shares},
  journal      = {J. {ACM}},
  volume       = {65},
  number       = {2},
  pages        = {8:1--8:27},
  year         = {2018}
}

@inproceedings{DBLP:conf/aaai/AzizRSW17,
  author       = {Haris Aziz and
                  Gerhard Rauchecker and
                  Guido Schryen and
                  Toby Walsh},
  title        = {Algorithms for Max-Min Share Fair Allocation of Indivisible Chores},
  booktitle    = {{AAAI}},
  pages        = {335--341},
  publisher    = {{AAAI} Press},
  year         = {2017}
}

@article{DBLP:journals/teco/ChakrabortyISZ21,
  author       = {Mithun Chakraborty and
                  Ayumi Igarashi and
                  Warut Suksompong and
                  Yair Zick},
  title        = {Weighted Envy-freeness in Indivisible Item Allocation},
  journal      = {{ACM} Trans. Economics and Comput.},
  volume       = {9},
  number       = {3},
  pages        = {18:1--18:39},
  year         = {2021}
}

@article{guo2023survey,
  title={A survey on fair allocation of chores},
  author={Guo, Hao and Li, Weidong and Deng, Bin},
  journal={Mathematics},
  volume={11},
  number={16},
  pages={3616},
  year={2023},
  publisher={MDPI}
}

@article{DBLP:journals/sigecom/AzizLMW22,
  author       = {Haris Aziz and
                  Bo Li and
                  Herv{\'{e}} Moulin and
                  Xiaowei Wu},
  title        = {Algorithmic fair allocation of indivisible items: a survey and new
                  questions},
  journal      = {SIGecom Exch.},
  volume       = {20},
  number       = {1},
  pages        = {24--40},
  year         = {2022}
}

@inproceedings{DBLP:conf/ijcai/00020L24,
  author       = {Fangxiao Wang and
                  Bo Li and
                  Pinyan Lu},
  title        = {Improved Approximation of Weighted {MMS} Fairness for Indivisible
                  Chores},
  booktitle    = {{IJCAI}},
  pages        = {3014--3022},
  publisher    = {ijcai.org},
  year         = {2024}
}

@article{DBLP:journals/mp/ShmoysT93,
  author       = {David B. Shmoys and
                  {\'{E}}va Tardos},
  title        = {An approximation algorithm for the generalized assignment problem},
  journal      = {Math. Program.},
  volume       = {62},
  pages        = {461--474},
  year         = {1993}
}

@article{hochbaum1988polynomial,
  title={A polynomial approximation scheme for scheduling on uniform processors: Using the dual approximation approach},
  author={Hochbaum, Dorit S and Shmoys, David B},
  journal={SIAM journal on computing},
  volume={17},
  number={3},
  pages={539--551},
  year={1988},
  publisher={SIAM}
}


\newpage

\appendix
\section*{Appendix}

\section{Missing Proofs in Section \ref{sec:can}}

\subsection{Proof of Theorem \ref{the:can}}
\label{prof:sec:can}

To prove Theorem~\ref{the:can},
we first decompose the definition of canonical instances into several properties, and show how to gradually modify an arbitrary instance $\cI=(\cN,\cM,\fw,\fv)$ to satisfy each of these properties one by one:
\begin{align*}
    \cI \to \cI^1 \to \cI^2 \to \cI^3 \to \cI^4 \to \cI^5.
\end{align*}
    \begin{enumerate}[label={(\arabic*)}]
        \item $\cI^1$: For each $a_i \in \cN$, $w_i=\frac{1}{2^p}w_1$, where $p$ is a non-negative integer.
        \item $\cI^2$: For each $a_i \in \cN$, $\WMMS_{i}=w_i \cdot \frac{v_{i}(\cM) }{\sum_{l\in [n]}w_l}$.
        \item $\cI^3$: For each $a_i \in \cN$, $v_{i}(\cM)=\sum_{l\in [n]}w_l=1$.
        \item $\cI^4$: For each $a_{i}\in \cN$ and each $e\in \cM$, $v_{i}(e)=\frac{1}{2^q}w_1$, where $q$ is a non-negative integer.
        \item $\cI^5$: For each $a_{i}\in \cN$, $v_{i}(e_1)\geq v_{i}(e_2)\geq \cdots \geq v_{i}(e_{m})$.
    \end{enumerate}
Eventually, $\cI^5$ is a canonical instance.
We will show that if $\cI^5$ has an $\alpha$-WMMS allocation, then $\cI$ has a $4\alpha$-WMMS allocation.

\paragraph{$\cI \to \cI^1$}
{First}, given an arbitrary instance $\cI=(\cN,\cM,\fw,\fv)$, we modify the weights of the agents and obtain $\cI^1=(\cN,\cM,\fw',\fv)$ satisfying property (1).
In particular, the weight $w_i$ of every agent $a_i$ is rounded up to the nearest power-of-$\frac{1}{2}$ fraction of $w_1$.
Formally, let $p = \lfloor \log_{\frac{1}{2}}\left(\frac{w_i}{w_1}\right) \rfloor$, and define the new weight as $w'_i = \frac{w_1}{2^p}$, where $w'_i$ is the weight of agent $a_i$ in $\cI^1$.
Accordingly, $w_i \le w'_i<2w_i$.
We claim that the WMMS of each agent increases by at most a factor of two.
\begin{claim}
    $\WMMS_{i}(\cI^1)\leq 2\WMMS_{i}(\cI)$ for all agents $a_i \in \cN$.
\end{claim}
\begin{proof}
    Let $\B=(B_1,\ldots,B_n)$ be a WMMS partition of agent $a_i$ in $\cI$.
    Then,  
    \[
    \WMMS_{i}(\cI^1)\leq w'_i\cdot \max_{a_{j} \in N} \frac{v_{i}(B_{j})}{w'_{j}} \leq \max_{a_{j} \in \cN}\frac{\WMMS_{i}(\cI) \cdot w_{j}\cdot w'_i}{w_i\cdot w'_{j}} \leq 2\WMMS_{i}(\cI),
    \]
    as $w'_i<2w_i$ and $w'_j\ge w_j$.
	    \end{proof}
	    Consequently, if agent $a_{i}$ receives a bundle $A_{i}$ such that $v_{i}(A_{i})\leq \alpha \WMMS_{i}(\cI^1)$ in instance $\cI^1$, it follows that $v_{i}(A_{i})\leq 2\alpha \cdot \WMMS_{i}(\cI)$ in the original instance $\cI$.

    \paragraph{$\cI^1 \to \cI^2$}
    Secondly, we transform $\cI^1=(\cN,\cM,\fw,\fv)$ satisfying property (1) to an instance $\cI^2=(\cN,\cM',\fw,\fv')$ satisfying both properties (1) and (2), where $\cM' = \cM\cup \cM^2$ and $\cM^2$ is a set of $n^2$ additional items that are introduced in this step.
    The agents' values for the original items in $\cM$ {remain} unchanged, i.e., $v'_i(e) = v_i(e)$ for all $a_i\in \cN$ and $e\in \cM$.
    For each agent $a_{i}\in \cN$, let $\B^{i}=(B^i_1,\ldots,B^i_n)$ be an arbitrary WMMS-defining partition in $\cI^1$ and {let $\cM_i^2 = \{e^i_1,\ldots,e^i_n\}$ be} a set of additional items.
    For each $e^i_j \in \cM_i^2$, let    
    \[
        v'_{i}(e^i_j)= \frac{\WMMS_i(\cI^1)}{w_i}\cdot w_j - v_{i}(B^i_j).
    \]
    Let $v_i(e^{i'}_j) = 0$ for all $i'\neq i$ and $j\in [n]$.
    Let $\cM^2=\bigcup_{i=1}^n \cM^2_i$.
    Thus, by the design of the valuations, for agent $a_i$, 
    \[
    \frac{v'_{i}(B^{i}_{1}\cup \{e^{i}_{1}\})}{w_1}=\cdots=\frac{v'_{i}(B^{i}_{n}\cup \{e^{i}_{n}\})}{w_n}=\frac{\WMMS_i(\cI^1)}{w_i}.
    \]
	    That is, $B^i_1\cup\{e^i_1\},\ldots,B^i_n\cup\{e^i_n\}$, with items in $\cM^2\setminus \cM^2_i$ arbitrarily distributed, also form a WMMS partition of agent $a_i$ in $\cI^2$, i.e., $\WMMS_i(\cI^2) = \WMMS_i(\cI^1)$, and 
    \[
    \frac{v'_i(\cM')}{\sum_{j\in [n]} w_j} = \frac{\WMMS_i(\cI^2)}{w_i}.
    \]
    Thus 
    \[
    \WMMS_i(\cI^2) = w_i\cdot\frac{v'_i(\cM')}{\sum_{j\in [n]} w_j}.
    \]


    Consequently, if agent $a_{i}$ is allocated a bundle $A_{i}$ such that $v'_{i}(A_{i})\leq \alpha \WMMS_{i}(\cI^2)$ in $\cI^2$, then 
    \[
    v_{i}(A_{i}\setminus \cM^2) \leq v'_{i}(A_{i})\leq \alpha \cdot \WMMS_{i}(\cI^2)=\alpha \cdot \WMMS_{i}(\cI^1).
    \]
    To summarize, we have the following claim.
    \begin{claim}
        $v'_{i}(A_{i})\leq \alpha \WMMS_{i}(\cI^2)$ implies $v_{i}(A_{i}\setminus \cM^2) \leq \alpha \WMMS_{i}(\cI^1)$.
    \end{claim}

    \paragraph{$\cI^2 \to \cI^3$}
    {Third}, we transform instance $\cI^2 =(\cN,\cM,\fw,\fv)$ satisfying properties (1)-(2) into instance $\cI^3=(\cN,\cM,\fw',\fv')$ satisfying properties (1)-(3).
	    {This step is simple: scale the weight of each agent by $\sum_{l\in [n]}w_l$, and scale the value of each item for agent $a_{i}$ by $v_{i}(\cM)$.}
    It is easy to verify that if agent $a_{i}$ is allocated a bundle $A_{i}$ such that $v'_{i}(A_{i})\leq \alpha \WMMS_{i}(\cI^3)$ in $\cI^3$, then $v_{i}(A_{i})\leq \alpha \WMMS_{i}(\cI^2)$ holds in the original instance $\cI^2$.

    \begin{claim}
        $v'_{i}(A_{i})\leq \alpha \WMMS_{i}(\cI^3)$ if and only if $v_{i}(A_{i})\leq \alpha \WMMS_{i}(\cI^2)$.
    \end{claim}

   \paragraph{$\cI^3 \to \cI^4$}
   Next, we transform instance $\cI^3=(\cN,\cM,\fw,\fv)$ satisfying properties (1)-(3) into instance $\cI^4=(\cN,\cM',\fw,\fv')$ satisfying properties (1)-(4), where $\cM' = \cM\cup \cM^4$ and $\cM^4$ is a set of additional {items} added in this step.

  For each agent $a_{i}$, let $\B^{i}=(B^i_1,\ldots,B^i_n)$ denote an arbitrary WMMS-defining partition of $a_i$.
  For each bundle $B^{i}_{j}$ and $e\in B^i_j$, we round the value $v_i(e)$ down to the nearest power-of-$\frac{1}{2}$ fraction of $w_1$.
  Formally, define $q=\lceil \log_{\frac{1}{2}}{\left(\frac{v_{i}(e)}{w_1}\right)} \rceil$ and set $v'_{i}(e)=\frac{1}{2^q}\cdot w_1$, and thus $\frac{1}{2}v_{i}(e) < v'_{i}(e)\leq v_{i}(e)$.
	  Since instance $\cI^3$ satisfies properties (1)-(3), $v_{i}(B^{i}_{j})=w_{j}=\frac{w_1}{2^p}$ for some integer $p$. 
	  Thus, there are $l$ positive integers $q_1, \ldots, q_l$ such that
	  \[
	  v_{i}(B^{i}_{j})- v'_{i}(B^{i}_{j}) = \sum_{t=1}^{l}\frac{w_1}{2^{q_t}}.
	  \] 
  We construct $l$ additional items, denoted by $\cM^4_{i,j}=\{e^{i,j}_1,\ldots,e^{i,j}_l\}$, such that $v'_{i}(e^{i,j}_{t})= \frac{w_1}{2^{q_t}}$ for $1\le t\le l$.
  The other agents have zero value for these items, i.e., $v'_{i'}(e) = 0$ for all $e\in \cM^4_{i,j}$ where $i\neq i'$.
  
  Let $\cM^4 = \bigcup_{i=1}^n \bigcup_{j=1}^{n} \cM^4_{i,j}$.
  Note that the total value of agent $a_{i}$ in $\cI^4$ {is the same as in} $\cI^3$, i.e., $v'_{i}(\cM\cup\cM^4)=v_{i}(\cM)$.
	  Moreover, the WMMS values are also preserved, i.e.,  $\WMMS_{i}(\cI^4)=\WMMS_{i}(\cI^3)$. 
  {To see this, fix an agent $a_i$ and consider the partition obtained from $\B^i$ by adding the compensating items in $\cM^4_{i,j}$ to the bundle $B^i_j$ for each $j$; compensating items created for other agents can be distributed arbitrarily, since $a_i$ assigns value zero to them. By construction,
  \[
  v'_i(B^i_j\cup \cM^4_{i,j})
  =v'_i(B^i_j)+\sum_{t=1}^{l}\frac{w_1}{2^{q_t}}
  =v_i(B^i_j)=w_j.
  \]
  Hence the compensated partition witnesses the same weighted maximin-share value for agent $a_i$ after the rounding step.}
  Thus, on top of properties (1)-(3), $\cI^4$ also satisfies (4).
  Consequently, if agent $a_{i}$ is allocated a bundle $A_{i}$ such that $v'_{i}(A_{i})\leq \alpha \WMMS_{i}(\cI^4)$, it follows that 
   \[
   v_{i}(A_{i}\setminus \cM^4)\leq 2v'_{i}(A_{i}\cap \cM) \leq 2\alpha \cdot \WMMS_{i}(\cI^4) =2\alpha \cdot \WMMS_{i}(\cI^3).
   \]

   To summarize, we have the following claim.

  \begin{claim}
      $v'_{i}(A_{i})\leq \alpha \WMMS_{i}(\cI^4)$ implies $v_{i}(A_{i}\setminus \cM^4)\leq 2\alpha \WMMS_{i}(\cI^3)$.
  \end{claim}


  
   \paragraph{$\cI^4 \to \cI^5$}
	   Finally, it has been proved {for example, in} \cite{DBLP:journals/aamas/BouveretL16,DBLP:journals/teco/BarmanK20,DBLP:conf/nips/0037WZ23}, that any instance $\cI^4$ can be transformed into an IDO instance $\cI^5$ in polynomial time, and if we can find an $\alpha$-WMMS allocation for instance $\cI^5$, we can find an $\alpha$-WMMS allocation for the original instance $\cI^4$.

   Combining the previous claims, we can prove the theorem. 
   \begin{proof}[Proof of Theorem \ref{the:can}]
   Given an arbitrary instance $\cI=(\cN,\cM,\fw,\fv)$, we can construct a canonical instance $\cI^5=(\cN,\cM',\fw',\fv')$, through $\cI^1, \cI^2, \cI^3, \cI^4$ in the above five steps.
   Suppose $(A_1,\ldots, A_n)$ is an $\alpha$-WMMS allocation in instance $\cI^5$, i.e., $v'_i(A_i) \le \alpha \WMMS_i(\cI^5)$, then
   \[
   v_i(A_i\setminus \cM^4) \le 2\alpha \WMMS_i(\cI^3) = 2\alpha \WMMS_i(\cI^2),  
   \]
   and
   \[
   v_i(A_i\setminus \cM^4 \setminus \cM^2) \le 2\alpha \WMMS_i(\cI^1) \le 4\alpha \WMMS_i(\cI), 
   \]
   which completes the proof.
   \end{proof}




\subsection{Proof of Theorem \ref{thm:polytime}}
\label{app:poly:proof}

\begin{proof}[Proof of Theorem \ref{thm:polytime}]
The proof is almost identical to that in \cite{DBLP:conf/ijcai/0001C019}, which relies on the rounding technique of \cite{DBLP:journals/mp/ShmoysT93}.
The key distinction in our proof is the observation that the values of 
$\WMMS_i$ can be computed via a PTAS, as shown in \cite{hochbaum1988polynomial}.
Since the result in \cite{DBLP:conf/ijcai/0001C019} does not exactly focus on the $\alpha$-WMMS problem, we include a description of the algorithm below for completeness.

    First, observe the following linear program:
    \begin{align*}
        \sum_{e_j\in\cM}v_i(e_j)\cdot x_{ij} &\leq \alpha\WMMS_i, &\forall i\in[n]\\
        \sum_{i\in [n]:v_i(e_j)\leq \WMMS_i}x_{ij} &=1, &\forall j\in [m]\\
        x_{ij} &\geq 0, &\forall i\in [n], j\in [m].
    \end{align*}

    A feasible solution of the above linear program can be viewed as a fractional $\alpha$-WMMS allocation, where $x_{ij}$ is the fraction of item $e_j$ assigned to agent $i$. 
    Note that this allocation is valid since each item is fully allocated, and each agent can only get a non-negative fraction of an item.
	    Because there exists an integral $\alpha$-approximate WMMS allocation, this program has at least one feasible solution. 
	    Then, we can find a basic feasible solution to this program and get an $\alpha$-WMMS fractional allocation in polynomial time.
	    By rounding this fractional solution to an integral one, we can obtain a feasible allocation for the instance. Note that this linear program is in the form of a generalized assignment problem; hence, we can apply standard rounding techniques \cite{DBLP:journals/mp/ShmoysT93} to find a $2\alpha$-WMMS allocation in polynomial time. 

    Nevertheless, although we can obtain an allocation in polynomial time from the above program, specifying the program itself proves difficult, because we need to know the WMMS value of each agent, which is intractable.
    
    Finding the WMMS value for an agent $i$ is equivalent to computing the {minimum makespan} on related machines in a job scheduling problem.
    Each item $e\in \cM$ can be viewed as a job with workload $v_i(e)$, and each agent with weight $w$ can be viewed as a machine with speed $w$.
    The total workload of a machine is the sum of the workload of jobs assigned to this machine, divided by its speed.
    The goal of this job scheduling problem is to assign all jobs to machines while minimizing the {maximum workload} of any machine (the makespan).
	    If we use $A$ to denote a job assignment ($A_j$ to denote the set of jobs assigned to machine $j$) and use $\bA$ to denote the set of all possible job assignments, {finding a schedule with minimum makespan amounts to computing} $ \min_{A \in \bA} \max_{j\in[n]} v_i(A_j)$, which happens to be $\WMMS_i$. This scheduling problem is known to be {NP-complete}.

	    Fortunately, this problem possesses a polynomial-time approximation scheme (PTAS) \cite{hochbaum1988polynomial}. Specifically, for an arbitrarily small $\epsilon >0$, one can find, in polynomial time, an $(1+\epsilon)$-approximation of the {minimum} makespan, which equals an $(1+\epsilon)$-approximation of $\WMMS_i$. Consequently, given a fixed parameter $\epsilon$, one can specify and solve the following linear program in polynomial time:

    \begin{align*}
        \sum_{e_j\in\cM}v_i(e_j)\cdot x_{ij} &\leq \alpha(1+\frac{\epsilon}{2\alpha})\WMMS_i, &\forall i\in[n]\\
        \sum_{i\in [n]:v_i(e_j)\leq (1+\frac{\epsilon}{2\alpha})\WMMS_i}x_{ij} &=1, &\forall j\in [m]\\
        x_{ij} &\geq 0, &\forall i\in [n], j\in [m].
    \end{align*}
    
    Therefore, by rounding the solution of this program, we can find a $(2\alpha+\epsilon)$-WMMS allocation for any instance in polynomial time.
\end{proof}

\section{Missing Proofs in Section \ref{sec:main}}

\label{sec:appendix:3WMMS}

\subsection{Proof of Claim \ref{cla:lr}}

\begin{proof}[Proof of Claim \ref{cla:lr}]
    {We prove the claim by contradiction.}
    Assume there is an agent $a_{i,j}$ and an item $e_h$ with $h> L_r$ such that $v_{i,j}(e_h) > w_r$. 
    {Since the instance is canonical, we have $v_{i,j}(e_1)\ge\cdots\ge v_{i,j}(e_h)\geq 2w_r$.}
    Let $\B = (B_{p,q})$ be an arbitrary WMMS-defining partition for $a_{i,j}$.
    By the definition, $v_{i,j}(B_{p,q}) \leq w_{p}$ for all bundles $B_{p,q} \in \B$, and thus the first $h$ items cannot be in any $B_{p,q}$ with $p\geq r$.  
    {Hence, $r\neq 1$; otherwise, item $e_h$ cannot belong to any bundle of $\B$, a contradiction.}
    Let $\overline{B} = \bigcup_{t=1}^{r-1}\bigcup_{p=1}^{n_t}B_{t,p}$ denote the union of all bundles corresponding to agents in groups $G_1$ to $G_{r-1}$, and we have $e_t\in \overline{B}$ for all $t\leq h$. 
    Therefore,
    \[  
        v_{i,j}(\overline{B}) \geq \sum_{t=1}^h v_{i,j}(e_t) \geq \frac{\sum_{t=1}^{r-1}W_t }{w_r} \cdot 2w_r = 2\sum_{t=1}^{r-1}W_t .
    \]
    This leads to a contradiction because $\B$ being a WMMS partition requires that
    \[
        v_{i,j}(\overline{B}) = \sum_{t=1}^{r-1}\sum_{q=1}^{n_t}v_{i,j}(B_{t,q}) \leq \sum_{t=1}^{r-1}\sum_{q=1}^{n_t}w_t \leq \sum_{t=1}^{r-1}W_t,
    \]
    which proves the claim.
\end{proof}

\subsection{Proof of Claim \ref{cla:yi}}

\begin{proof}[Proof of Claim \ref{cla:yi}]
{We first label the items in $[1,L_{r_2}]$ and assign costs to them as in the proof of Claim~\ref{cla:y1}.}
Then we label (or relabel) and assign (or change) costs to the items in $S_{j} = [{d_{r_{j}}+1},{L_{r_{j+1}}}]$ for $2\leq j\leq i$.

{For each $j=2,\ldots,i$, perform one round of the following process.}
{In round $j$, every item $e\in S_j$ is assigned cost $v^{\star}_{r_j}(e)$.}
{The items in $[{d_{r_j}+1},{L_{r_j}}]$ are labeled $\kappa$; denote this set by $S^1_{j,\kappa}$.}
{No item in $S^1_{j,\kappa}$ is relabeled later, and its assigned cost is not changed.}
Similar to the proof of Claim~\ref{cla:y1}, the items $[{L_s+1},{L_s+n_s}]$ are labeled with $\gamma$ and denoted as $S^{s}_\gamma$ for $s \in \{r_{j}, r_{j}+1, \ldots, r_{j+1}-1\}$.
The remaining items in $S_{j}$ are labeled with $\lambda$.

Now every item in $[d_{r_j}+1,L_{r_{j+1}}]$ is given a cost depending on the valuation $v^{\star}_{r_j}(\cdot)$.
But we still need to relabel some items.
{In the $j$-th round, by the construction of $v^{\star}_{r_j}(\cdot)$, every item labeled $\lambda$ with index larger than $L_{r_j}$ is assigned a cost equal to the weight of some agent.}
{We match each such item to the corresponding agent without redundancy; we also apply this step to the items in $[1,d_{r_2}]$ labeled $\lambda$.}
{Here, the corresponding agent means an agent whose weight equals the cost assigned to the item in the preceding sentence.}
{Since each item labeled $\lambda$ in $[1,L_{r_j}]$ was matched to an agent in an earlier round, if an item $e$ is matched to an agent $a_u$ that has already been matched to another item in $[1,L_{r_j}]$, we relabel $e$ with $\kappa$.}
{By Claim~\ref{cla:lr}, for any agent with weight $w_s<w_{r_j}$ that is first matched with an item in this round, at most $s-(j+1)$ items are matched with this agent in the following rounds.}
{This step also ensures that the number of items labeled $\lambda$ with cost $w_s$ is at most $n_s$.}

The round of process is repeated until all the items in $[1,{L_i}]$ have been labeled and assigned costs.





{Only items in $[1,n_1]\cup \cdots \cup [L_{r_{i+1}-1}+1, L_{r_{i+1}-1}+n_{r_{i+1}-1}]$ can be labeled $\gamma$.}
The cost of any item in $[L_s+1,L_s+n_s]$ is at most $w_s$.
Hence, the total cost of the items with label $\gamma$ in $[1,\ldots, L_{r_{i+1}}]$ is at most $W_1+\ldots+W_{r_{i+1}-1}$.

Denote by $S^s_\lambda$ the set of items in $[1, L_{r_{i+1}}]$ that have label $\lambda$ and cost $w_s$, for $s \in [r_i]$. 
\acomment{Author verification: This corresponds to the Lean predicate \texttt{LambdaLabelConstructed} and theorem \texttt{proxyInductionLambdaConstructed\_of\_article\_setup}; it is also part of the verification markers \texttt{InductionLabelingDefinedByArticle} and \texttt{ArticleLabelingDefinedByArticle}. The formalization uses an explicit endpoint for the $\lambda$-label levels because this sentence defines $S^s_\lambda$ for $s\in[r_i]$, while the subsequent accounting bound runs through $r_{i+1}-1$. Please confirm that the intended range is all relevant levels up to $r_{i+1}-1$, or adjust the range notation.}\hresp{Tentatively confirmed with low confidence. The intended range should likely include all relevant $\lambda$-label levels up to $r_{i+1}-1$; we will revisit whether this is sufficient during the later Lean verification pass.}
By Claim~\ref{cla:lr}, we have $|S^s_\lambda| < L_s$.
{By the matching process, at most $n_s$ items with cost $w_s$ are labeled $\lambda$.}
Let $X_s$ denote the total cost of $S^s_\lambda$, and set $\bar{X}_s = W_1 + \cdots + W_{s-1} - X_s$. 
For any $2 \leq s \leq r_i - 1$, we have $X_s < \bar{X}_{s-1} + 2W_s$, and $L_{r_i} \cdot w_{r_i} \leq \bar{X}_{r_i-1} + 2W_{r_i}$.
Hence, the total cost of the items with label $\lambda$ is no more than $W_1 + 2W_2 + \cdots + 2W_{r_{i+1}-1}$.

Now we bound the total cost of items with label $\kappa$ in $[1,L_{r_{i+1}}]$.
For any $1\leq j \leq i-1$, let $Z^{1}_{j} = {(L_{r_{j+1}} - d_{r_{j+1}})} \cdot w_{r_{j+1}}$, which is the total cost of $S^{1}_{j+1,\kappa}$.
{Let $S^{2,+}_{j,\kappa}$ denote the items in $[{d_{r_j}+1},{L_{r_{j+1}}}]$ that have label $\lambda$ and cost less than $w_{r_{j+1}}$, and let $Z^2_j$ be the nonnegative increase obtained by replacing each cost in $S^{2,+}_{j,\kappa}$ by $w_{r_{j+1}}$:
\[
Z^2_j=|S^{2,+}_{j,\kappa}|\cdot w_{r_{j+1}}-\sum_{e\in S^{2,+}_{j,\kappa}}v(e),
\]
}
By inequality~\ref{eq:i'}, we have $Y_j \geq Z^1_j + Z^2_j$.

Denote all items labeled $\kappa$ as $S_\kappa$, and let $S^2_\kappa = S_\kappa \setminus \left(\bigcup_{2 \leq j \leq r_i} S^{1}_{j,\kappa}\right)$.
{An item with cost $w_s$ in $S^2_\kappa$ is labeled $\kappa$ because it and an item in some $S^{2,+}_{j,\kappa}$ with the same cost are matched with the same agent $a$.}
We have ${w_{r_{j+1}} - w_s = (\log_2\left(\frac{w_{r_{j+1}}}{w_s}\right)-1)\cdot w_s} \geq (s - (j+1))\cdot w_s$.
The value of ${w_{r_{j+1}}-w_s}$ is at least the total cost of the items with label $\kappa$ in $S^2_\kappa$ which are matched with agent $a$.
The total cost of $S^2_\kappa$ can be covered by $Z^2_1+\ldots+Z^2_{i-1}$.
Therefore, the total cost of $S_\kappa$ is no larger than $\sum_{j \in [i-1]} Y_j$.

Each item in $[1,L_{r_{i+1}}]$ is labeled with $\gamma$, $\lambda$ or $\kappa$.
{Using the bounds for the three labels, we obtain}
\[
\sum_{j \in [i-1]} v^{\star}_{r_j}\left([d_{r_j}+1, d_{r_{j+1}}]\right)_{w_{r_{i+1}}} 
+ v^{\star}_{r_i}\left([d_{r_i}+1, L_{r_{i+1}}]\right)_{w_{r_{i+1}}}
- \sum_{j \in [i-1]} Y_j 
< 
3W_1 + \cdots + 3W_{r_{i+1}-1}.
\]

\end{proof}

\section{Missing Proofs in Section \ref{sec:lower}}
\label{sec:appendix:lb}

\subsection{Proof of Theorem \ref{the:lb}}

\begin{proof}[Proof of Theorem \ref{the:lb}]
{We prove by contradiction. Assume} that there exists an allocation $\A=(A_{1,1},\ldots,A_{k,n_k})$ of $\overline{\cI}$ whose approximation ratio $\alpha <2-\epsilon$.
Observe that, in $\A$, no agent in $G_i$ can obtain any item in $B_i$, since each item in $B_i$ has cost $2w_i$ to them, and allocating any single item breaks the approximation of 2.
Therefore, we have the following simple observation.


\begin{observation}
For each $i\in[k-1]$, all the items in $B_{i+1}$ are allocated to the agents in $G_1\cup \cdots \cup G_i$.
\end{observation}

Before applying other changes to $\A$, we first transform it into a \emph{fractional} allocation $\X$.
As the name implies, in a fractional allocation, one item can be split and allocated to more than one agent. 
Denote by $0\le f(a,e)\le 1$ the fraction of item $e$ allocated to agent $a$. We still need to ensure that $\sum_{a\in\cN}f(a,e)=1$ for each $e\in \cM$; that is, each item is fully allocated and no items are allocated excessively.
Given a fractional allocation, an agent's cost is $\sum_{e\in \cM}v_i(e)f(a_i,e)$.
{An integral allocation is a special case of a fractional allocation,}
and the transformation {from $\A$ to a fractional allocation $\X$} itself has no impact on agents' costs.
We apply this generalization in order to simplify the following discussion and prove that no fractional allocation $\X$ can be better than $(2-\epsilon)$-WMMS, nor can the integral allocation $\A$.

We next introduce two local transformations of $\X$ that do not increase any agent's cost.


\paragraph{ Property 1} If there exists a (fraction of) item $e\in T'_i$ that is allocated to an agent $a_{i',j}$ not in $G_i$, reallocating it to any agent in $G_i$ does not increase any agent's cost.

\medskip

For agent $a_{i',j}$, her cost will decrease after removing an item, and for the agent receiving item $e$, she has cost $0$, so her cost will also not increase, and Property 1 is trivially true.

\paragraph{Property 2} Given two agents $a_{i,j},a_{i+1,j'}$ and two items $e \in A_{i,j}, e'\in A_{i+1,j'}$, if $e\notin \bigcup_{t=1}^{i+1}\cM_t$ and $e'\in \cM_{i+1}^2$, swapping a fractional assignment of $e$ and $e'$ so that either $f(a_{i,j},e)=0$ or $f(a_{i+1,j'},e')=0$ will not increase the cost of any agent.

\medskip

First note that because $e\notin \bigcup_{t=1}^{i+1}{\cM_t}$, both agents have the same value for $e$.
{Second, observe that $v_{i+1}(e') = w_{i+1} > \frac{1}{2}w_{i+1} = v_i(e')$.}
Therefore, we can perform the following: Agent $a_{i,j}$ selects a small fraction of $e$ with cost $x$ from her bundle and {moves} this fraction to $A_{i+1,j'}$. Then agent $a_{i+1,j'}$ chooses a fraction of $e'$ with cost $x$ and gives it to $a_{i,j}$. After the exchange of items, agent $a_{i+1,j'}$ has {unchanged total cost}, and agent $a_{i,j}$ ends up with a lower cost since the fraction of $e'$ she receives from $a_{i+1,j'}$ is worth only $\frac{w_{i+1}}{2}$ from her perspective. Hence, we can repeat this process unless either $e$ or $e'$ is exhausted.

\medskip

Given the above two properties, without loss of generality, we have the following observation about $\X$.

\begin{observation}
    As long as $f(e,a)>0$ for some $e\in \cM_{i+1}^2$ and $a\in G_{i+1}$, for any agent $a' \in G_i$ and $e' \notin \bigcup_{t=1}^{i+1}\cM_t$, $f(a',e')=0$.
\end{observation}

Next, we are going to enhance this observation and show that its ``if'' condition always holds true.
To simplify the expression, we use {$\beta_i$} to represent the percentage of items in $\cM_{i+1}^2$ that are allocated to agents in $G_i$, i.e.,
\[
\beta_i = \frac{\sum_{a\in G_i, e\in \cM_{i+1}^2} f(a,e)}{|\cM_{i+1}^2|}.
\]

\begin{claim}\label{cla:opt1}
    $\beta_i <1$ for all $i\in[k-1]$.
\end{claim}
\begin{proof}
    We prove this claim by induction.
    First, consider $G_1$. Observation 1 requires that all items in $B_2$ are allocated to $G_1$. Assume $\beta_1 \le 1$.
    
    Then, we can estimate the worst approximation ratio among agents in $G_1$ (albeit there is only one agent) by the average, that is, the cost of all items allocated to $G_1$ divided by the total weight of agents in $G_1$:
    \begin{align*}
        \frac{v_1(B_2) + \beta_1v_1(\cM_2^2)}{W_1}
        =&\frac{1+w_2(\frac{1}{2}n_2-m_1)+ \frac{\beta_1}{2}w_2 n_2}{1} \\
        = &\frac{3}{2}-\frac{w_2 m_1}{1}+\frac{\beta_1}{2}>\frac{3}{2}-\frac{1}{3n_2}+\frac{\beta_1}{2}.
    \end{align*}
    Note that we do not need to consider any item {outside} $\cM_1\cup \cM_2$.
     This is because, when $\beta_1 = 1$,
    \[
        \frac{3}{2}-\frac{1}{3n_2}+\frac{\beta_1}{2}>2-\frac{1}{3n_2} \geq 2-\frac{1}{3\cdot 2^k}\geq 2-\frac{1}{3\cdot k^2}> 2-\frac{\epsilon^4}{48}>2-\epsilon,
    \]
    which contradicts the assumption that $\X$ performs better than $(2-\epsilon)$-WMMS.

    Note that the items in $\cM_2^2$ can only be allocated to agents in $G_1$ or $G_2$ {by} Observation 1.
    Then, assume that the claim holds for the first $p-1$ groups. Since $\beta_i <1$ for $i\in[p-1]$, no agents from the first $p-1$ groups are allocated any items in $T_{p+1}$. Therefore, Observation 1 implies that each item in $T_{p+1}$ is allocated to agents in $G_p$. 
    As such, the items allocated to $G_p$ can be divided into four parts: $1-\beta_{p-1}$ of the items in $\cM_p^2$, items in $T_{p+1}$, items in $\cM_{p+1}^2$, and other items given $\beta_p=1$. We apply a similar estimate of the WMMS ratio of agents in $G_p$, ignoring the potential last part of items:

    \begin{align*}
        & \frac{(1-\beta_{p-1})v_p(\cM_p^2)+v_p(T_{p+1}) + \beta_p v_p(\cM_{p+1}^2)}{W_p}\\
        =&\frac{(1-\beta_{p-1})2^{p-2}+w_{p+1}(\frac{1}{2}n_{p+1}-\sum_{q=1}^p m_q)+\frac{\beta_p}{2} w_{p+1}n_{p+1}}{2^{p-2}}\\
        \ge& {(1-\beta_{p-1})+1-\frac{2w_{p+1}n_{p}}{2^{p-2}}+\beta_p \geq 2 - \frac{4}{\Delta} + \beta_p - \beta_{p-1}}.
    \end{align*}

    {When $\beta_p=1$, the approximation ratio contradicts the assumption because $\frac{4}{\Delta}=\frac{4}{2^k}\le \frac{4}{k^2}< \frac{\epsilon^4}{4}<\epsilon$.}
\end{proof}

The estimates used in the above proof can also be used to capture the overall WMMS approximation ratio of $\X$. We can use the following linear program to show the maximum WMMS estimate among groups.
Denote $\alpha_i=\max_{a_{i,j}\in G_i} \frac{v_{i,j}(X_{i,j})}{w_i}$.
    \begin{align*}
      & \min \quad \max_{i\in [k]} \alpha_i \\
      & \begin{array}{r@{\quad}l@{}l@{\quad}l}
        s.t. &\alpha_1 &\ge  1 + (\frac{1}{2}-\frac{1}{3n_2}) + \frac{1}{2}\beta_1, \\
        &\alpha_2 &\ge 2 - \frac{4}{\Delta} + \beta_2 - \beta_1, \\
        &  &~\vdots\\
        &\alpha_{k-1} &\ge 2 - \frac{4}{\Delta} + \beta_{k-1} - \beta_{k-2}, \\
        & \alpha_{k} &\ge 1-\beta_{k-1},
       \end{array}\\
       &\quad\quad~~~ 0\leq \beta_i <1 \quad \text{for $i\in [k-1]$}.
    \end{align*}
Let
\(
c = 2-\frac{4}{\Delta}.
\)
For $i=2,\ldots,k-1$, the middle constraints give
\(
\alpha_i \ge c+\beta_i-\beta_{i-1}.
\)
Moreover,
\[
\sum_{i=2}^{k-1}(\beta_i-\beta_{i-1})=\beta_{k-1}-\beta_1>-1,
\]
because $\beta_{k-1}\geq 0$ and $\beta_1<1$. Hence it is impossible that all $k-2$ differences
$\beta_i-\beta_{i-1}$ are at most $-\frac{1}{k-2}$; otherwise the sum would be at most $-1$.
Therefore, for some $i\in\{2,\ldots,k-1\}$,
\(
\beta_i-\beta_{i-1}>-\frac{1}{k-2},
\)
and consequently
\(
\max_{i\in[k]}\alpha_i
>2-\frac{4}{\Delta}-\frac{1}{k-2}.
\)
Therefore, the WMMS approximation ratio is greater than
\[
2 - \frac{4}{\Delta} -\frac{1}{k-2}\geq 2-\frac{\epsilon^4}{4}-\frac{\epsilon^2}{4-\epsilon^2}>2-\epsilon.
\]

In conclusion, we have shown that the approximation ratio of $\X$ cannot be lower than $2-\epsilon$, which implies the impossibility of $\A$, contradicting the initial assumption. Thus, Theorem \ref{the:lb} is proved.
\end{proof}

\subsection{Proof of Theorem \ref{thm:lb:n}}

\begin{proof}[Proof of Theorem \ref{thm:lb:n}]
For any $n\ge 2$, we construct an instance $\cI=(\cN,\cM,\fw,\fv)$ with $\cN=\{a_1,\ldots,a_n\}$ and $\cM=\{e_1,\dots,e_{2n-1}\}$. 
{We divide the items into three sets, $M_1=\{e_1\}$, $M_2 = \{e_2, e_3,\ldots, e_{n}\}$, and $M_3=\{e_{n+1}, e_{n+2}, \ldots, e_{2n-1}\}$, and the items in each set are homogeneous.}
Let $\alpha=\frac{n+\sqrt{5n^2-4n}}{2n}$.
The valuations of the agents are shown in Table \ref{tab:lb:n:1}.
It can be verified that $\sum_{i=1}^n w_i = 1$, $v_1(\cM)=v_2(\cM) = 1$, $\WMMS_1=\frac{1}{\alpha}$ and $\WMMS_i=\frac{1}{n\alpha^2}$ for all $i\ge 2$.

\begin{table}[h]
    \centering
    \renewcommand{\arraystretch}{1.2}
    \begin{tabular}{c|c|ccc}
    \hline
       agent  & weight & $e\in M_1$ & $e\in M_2$ & $e\in M_3$  \\
       \hline
        $a_1$ & $\frac{1}{\alpha}$ & $\frac{1}{\alpha}$ & $\frac{1}{n\alpha^2}$ & 0\\
        \hline
        $a_i$ for $i\ge 2$ & $\frac{1}{n\alpha^2}$ & $\frac{1}{n\alpha}$ & $\frac{1}{n\alpha}$ & $\frac{1}{n\alpha^2}$\\
        \hline
    \end{tabular}
    \caption{The valuation of $a_1,\ldots,a_n$.}
    \label{tab:lb:n:1}
\end{table}

We consider two cases of allocating the items:
\begin{itemize}
    \item Case 1. If one item in $M_1\cup M_2$ is allocated to an agent $a_i$ for $i\ge 2$, the unfairness ratio for $a_i$ is at least $\alpha$. 
    \item Case 2. If all items in $M_1\cup M_2$ are allocated to agent $a_1$, the unfairness ratio for $a_i$ is at least $1+\frac{n-1}{n\alpha}=\alpha$.
\end{itemize}
Hence, for any allocation, the approximation ratio is no less than $\alpha$.
\end{proof}

\section{Missing Proofs in Section \ref{sec:two}}

\label{sec:appendix:two}


\subsection{Proof of Lemma \ref{lem:two:upper}}

\begin{proof}[Full Proof of Lemma \ref{lem:two:upper}]
    Given an instance with two agents $\{a_1,a_2\}$, denote by $\{B_1,B_2\}$ the WMMS-defining partition of agent $a_1$ and $\{C_1,C_2\}$ the WMMS-defining partition of agent $a_2$. 
    By the proportional proposition in Theorem \ref{the:can}, without loss of generality, it is assumed that $\WMMS_i=w_i$, and $v_1(B_1)=v_2(C_1)=w_1$ and $v_1(B_2)=v_2(C_2)=w_2$.
    We divide all items into 4 atomic bundles:
$X_{1,1}=B_1\cap C_1$, $X_{1,2}=B_1\cap C_2$, $X_{2,1}=B_2\cap C_1$, $X_{2,2}=B_2\cap C_2$.
Accordingly, there are coefficients $0\le a,b,c,d\le 1$ such that the values of the atomic bundles can be {represented} as in Table \ref{tab:parts}.
In the following, we use different combinations of atomic bundles to design the algorithm and allocate items.


\begin{table}[htbp]
    \centering
    \begin{tabular}{c|c|c|c|c}
        \hline
        agents  & $X_{1,1}=B_1\cap C_1$ & $X_{1,2}=B_1\cap C_2$ & $X_{2,1}=B_2\cap C_1$ & $X_{2,2}=B_2\cap C_2$ \\
        \hline
        $a_1$ & $aw_1$ & $(1-a)w_1$ & $dw_2$ & $(1-d)w_2$\\
        \hline
        $a_2$ & $bw_1$ & $(1-c)w_2$ & $(1-b)w_1$ & $cw_2$\\
        \hline
    \end{tabular}
    \caption{Values for {atomic bundles}.}
    \label{tab:parts}
\end{table}

Our algorithm is simple: 
\begin{itemize}
    \item For any two-agent instance, return one of the eight allocations in Table \ref{tab:alg} {that has} the smallest unfairness ratio.
\end{itemize}
It is somewhat surprising that this simple algorithm achieves the optimal WMMS approximation ratio for any weight distribution.


{\small
\begin{table}[htbp]
  \centering
  
    \begin{tabular}{c|c|c|c|c}
    \hline
    ID & $A_1$ & $A_2$ & $v_1(A_1)$ & $v_2(A_2)$ \bigstrut\\
    \hline
    1     & $X_{1,1}\cup X_{1,2}$ & $X_{2,1}\cup X_{2,2}$ & $w_1$ & \cellcolor[rgb]{ .757,  .941,  .784}$(1-b)w_1+cw_2$\bigstrut\\
    \hline
    2     & $X_{1,1}\cup X_{2,1}$ & $X_{1,2}\cup X_{2,2}$ & \cellcolor[rgb]{ .757,  .941,  .784} $aw_1+dw_2$ & $w_2$ \bigstrut\\
    \hline
    3     & $X_{1,1}\cup X_{2,2}$ & $X_{1,2}\cup X_{2,1}$ & \cellcolor[rgb]{ .757,  .941,  .784} $aw_1+(1-d)w_2$ & \cellcolor[rgb]{ .757,  .941,  .784} $(1-b)w_1+(1-c)w_2$ \\
    \hline
    4     & $X_{1,2}\cup X_{2,1}$ & $X_{1,1}\cup X_{2,2}$ & \cellcolor[rgb]{ .757,  .941,  .784} $(1-a)w_1+dw_2$ & \cellcolor[rgb]{ .757,  .941,  .784} $bw_1+cw_2$ \\
    \hline
    5     & $X_{1,2}\cup X_{2,2}$ & $X_{1,1}\cup X_{2,1}$ & \cellcolor[rgb]{ .757,  .941,  .784} $(1-a)w_1+(1-d)w_2$ & \cellcolor[rgb]{ .757,  .941,  .784} $w_1$ \\
    \hline
    6     & $X_{2,1}\cup X_{2,2}$ & $X_{1,1}\cup X_{1,2}$ & $w_2$ & \cellcolor[rgb]{ .757,  .941,  .784} $bw_1+(1-c)w_2$ \\
    \hline
    7     & $X_{1,1}\cup X_{1,2}\cup X_{2,2}$ & $X_{2,1}$ & \cellcolor[rgb]{ .757,  .941,  .784} $w_1+(1-d)w_2$ &  \cellcolor[rgb]{ .757,  .941,  .784} $(1-b)w_1$ \\
    \hline
    8     & $X_{1,2}\cup X_{2,1}\cup X_{2,2}$ & $X_{1,1}$ & \cellcolor[rgb]{ .757,  .941,  .784} $(1-a)w_1+w_2$ & \cellcolor[rgb]{ .757,  .941,  .784} $bw_1$ \\
    \hline
    \end{tabular}%
    \caption{Candidate allocations.}
  \label{tab:alg}%
\end{table}%
}

For each allocation in Table \ref{tab:alg}, we color the agent(s) whose bundles may reach the unfairness ratio in green.
The worst-case approximation ratio of our algorithm, denoted by $t$, is achieved for particular $0\le a,b,c,d\le 1$ such that all these allocations have unfairness ratios at least $t$. 
That is, the approximation ratio of our algorithm is the optimal value of the following program $\mathcal{LP}$:


\begin{gather*}
\begin{aligned}
    \max \quad t \\
    \textup{s.t.}\quad t &\leq (1-b)\frac{w_1}{w_2}+c\\
    t &\leq a+d\frac{w_2}{w_1}\\
    t&\leq \max\{a+(1-d)\frac{w_2}{w_1},1-c+(1-b)\frac{w_1}{w_2}\}\\
t&\leq \max\{1-a+d\frac{w_2}{w_1},b\frac{w_1}{w_2}+c\}\\
t&\leq \max\{(1-a)+(1-d)\frac{w_2}{w_1},\frac{w_1}{w_2}\}\\
 t&\leq b\frac{w_1}{w_2}+1-c\\
t&\leq \max\{1+(1-d)\frac{w_2}{w_1}, (1-b)\frac{w_1}{w_2}\}\\
 t&\leq \max\{1-a+\frac{w_2}{w_1},b\frac{w_1}{w_2}\}\\
 0&\leq a, b, c, d \leq 1
\end{aligned}
\end{gather*}

 
\paragraph{Mathematica computation.}
The following Wolfram Mathematica command was used to solve the symbolic program, where $x=\frac{w_1}{w_2}$. 
\begin{verbatim}
Simplify[
  Maximize[
    {
      t,
      t <= (1 - b) x + c &&
      t <= a + d/x &&
      t <= Max[a + (1 - d)/x, (1 - b) x - c + 1] &&
      t <= Max[-a + d/x + 1, b x + c] &&
      t <= x &&
      t <= b x - c + 1 &&
      t <= Max[(1 - d)/x + 1, (1 - b) x] &&
      t <= Max[-a + 1/x + 1, b x] &&
      0 <= a <= 1 &&
      0 <= b <= 1 &&
      0 <= c <= 1 &&
      0 <= d <= 1 &&
      x > 1
    },
    {a, b, c, d, t}
  ]
]
\end{verbatim}
The command returns that the optimal value $t^*$ of $\mathcal{LP}$ is
 \begin{equation*}  
t^*=\left\{  
     \begin{array}{ll}  
     0.5+0.5\frac{w_1}{w_2}, & 1\leq \frac{w_1}{w_2} < \frac{1}{2}(\sqrt{5}+1)\\
      1+0.5\frac{w_2}{w_1},  & \frac{1}{2}(\sqrt{5}+1)\leq \frac{w_1}{w_2} < \sqrt{2}+1\\
      0.5\frac{w_1}{w_2}, & \sqrt{2}+1\le \frac{w_1}{w_2}< \sqrt{3}+1\\
      1+\frac{w_2}{w_1}, & \sqrt{3}+1\le\frac{w_1}{w_2},
     \end{array}  
\right.  
\end{equation*} 
which is exactly $f(\alpha)$ with $\alpha = \frac{w_1}{w_2}$.
This completes the proof of the lemma.
\end{proof}

\end{sloppypar}
\end{document}